\documentclass[12pt]{article}
\usepackage[top=1in,bottom=1in,left=1in,right=1in]{geometry} 
\usepackage{palatino}

\usepackage[T1]{fontenc}
\usepackage{lmodern}
\usepackage[american]{babel}
\usepackage{tabularx}
\usepackage{booktabs}
\usepackage[table]{xcolor}

\usepackage{enumerate}
\usepackage{verbatim}
\usepackage{natbib}
\usepackage{amsmath}
\usepackage{subfigure}
\usepackage{enumerate}
\usepackage{subfigure}
\usepackage{verbatim}
\usepackage{natbib}
\usepackage{multirow}
\usepackage{graphicx}
\usepackage{amssymb,amsmath}
\usepackage{color}
\usepackage{url}

\normalsize

\def\be{\begin{eqnarray}}
\def\ee{\end{eqnarray}}
\usepackage{epsfig}

\newtheorem{thm}{Theorem}
\newtheorem{lem}{Lemma}

\title{A Likelihood Ratio Testing Approach for Interval-Censored Data}

\author{Yuan Wu$^\ast$, and Susan Halabi$^\dagger$}

\date{}

\begin{document}

\maketitle

\begin{center} 
Department of Biostatistics and Bioinformatics, Duke University,\\ Durham, North Carolina, USA\\[2pt] $^\ast$\texttt{yuan.wu@duke.edu}, $^\dagger$\texttt{susan.halabi@duke.edu} 
\end{center}

\begin{abstract}
Interval-censored data frequently arise in clinical research where event times are only known to fall within specific assessment windows. Although the Cox proportional hazards model is a standard approach for such data, existing Wald-type tests often suffer from instability or poor performance in small samples. In this paper, we propose a robust spline-sieve-based likelihood ratio test for interval-censored data. We develop a computationally efficient estimation framework that ensures numerical stability. Furthermore, we rigorously establish the asymptotic distribution of the proposed likelihood ratio statistic, providing a solid theoretical foundation for statistical inference. Extensive simulation studies demonstrate that our approach achieves superior error control and higher power compared with traditional approaches. The practical utility of the method is further illustrated through the analysis of a real-world clinical dataset.

\end{abstract}

\section{Introduction}
\label{sec1}
Accurate estimation of time-to-event outcomes is critical in oncology clinical trials, as these endpoints guide treatment decisions and ensure patients receive therapy at the most appropriate time.
In many phase II and III cancer trials, the primary endpoint is a time-to-event outcome such as progression-free survival (PFS), typically assessed every 6 months until disease progression~\cite[]{armstrong:2020,stadler:2006}.
Because the exact timing of progression is only known to lie between two assessments, such data are subject to interval-censoring. 

A common but problematic practice in clinical research is to simplify interval-censored data by converting them into right-censored data. Investigators then apply standard methods such as the Kaplan–Meier estimator~\cite[]{kaplan:meier:1958}, the log-rank test~\cite[]{mantel:1966}, or the Cox proportional hazards (PH) model~\cite[]{cox:1972}.  Specifically, when an event is known to occur only within an interval, investigators often assign a single time point (typically the right end of the interval, that is, the last observation time) as the event time, thereby converting the data into right-censored form. Although convenient, this approach can introduce bias~\cite[]{Turnbull:76}. Moreover, \cite{wu:halabi:2019} showed that, even when a random point within the interval is selected as the event time, the resulting estimate generally fails to attain the maximum likelihood estimate obtained under proper methods for interval-censoring. Fully parametric modeling approaches~\cite[]{hudgens:li:fine} that rely on distributional assumptions provide another straightforward option, but their validity hinges on correctly specifying the functional form of the target event-to time distribution. 

General methodologies for interval-censored data, encompassing nonparametric estimation, semiparametric regression, and hypothesis testing, are extensively reviewed by \cite{huang:wellner:1997} and \cite{sun:06}. Within this framework, the Cox proportional hazards model remains a cornerstone of clinical research. However, interval censoring precludes the use of the standard partial likelihood \cite{cox:1972}, necessitating full likelihood approaches. While nonparametric maximum likelihood estimation (MLE) \cite{zeng:mao:lin:16, zeng:gao:lin:17} and spline-sieve MLE \cite{zhang:Hua:Huang, wang:hudgens:2016} have been proposed for this context, these methods rely on the EM algorithm or constrained optimization, which can be challenging  to implement in practice. Specifically, these algorithms often suffer from computational instability and slow convergence, with performance issues becoming particularly pronounced  in the small-sample settings typical of phase II oncology trials.

In this paper, we develop a sieve semiparametric likelihood ratio method for inference under the Cox proportional hazards model with interval-censored data. Our approach offers two primary advantages. First, by leveraging the likelihood ratio principle, which is often more reliable than Wald-type tests in small samples~\cite[]{meeker:1995, murphy:vandervaart:1997}, we improve the robustness of the statistical inference. Second, we propose an efficient and stable computational strategy for estimating the sieve MLE, thereby addressing the practical limitations of existing  algorithms. Together, these innovations provide a practical and accurate approach for analyzing interval-censored survival data, ultimately strengthening the evidence base for clinical decision-making. In addition,  we establish the theoretical foundation for the proposed approach. Although the likelihood ratio statistic arises directly from the semiparametric MLE, deriving its asymptotic distribution under interval-censoring is a non-trivial task~\cite[]{murphy:vandervaart:1997}.

This paper is organized as follows. Section~\ref{sec2} introduces the construction of sieve semiparametric likelihood ratio statistics, based on sieve MLEs obtained from the full likelihood function and from likelihoods where subsets of the regression coefficients are treated as unknown. An efficient and robust computational method for obtaining the sieve MLE is also proposed. Section~\ref{sec3} establishes the asymptotic $\chi^2$ distribution of the likelihood ratio statistic, building on the convergence and normality of the sieve MLE. Section~\ref{sec4} presents simulation studies demonstrating that the proposed likelihood ratio test  outperforms existing Wald-type tests for the Cox PH model with interval-censoring. Section~\ref{sec5} applies the method to a real dataset for illustration. Section~\ref{sec6} discusses the practical value of the proposed approach and outlines future research directions. Technical details are presented in the Supplementary Material.  R scripts implementing the proposed methods are available at {\color{blue}\url{https://github.com/project-unknown-cloud/LR-IntervalCensoring}}.

\section{Likelihood Ratio Statistics and Their Computing}
\label{sec2}

Assume that the event time $T$ is interval-censored by $U$ and $V$ with $\Delta_1=1_{[T\le U]}$, $\Delta_2=1_{[U<T\le V]}$ and $\Delta_3=1-\Delta_1-\Delta_2$, $\boldsymbol{Z}=(Z_1,\cdots,Z_K)^\top$ represents the covariates vector with $K$ components. Let the observed data $\{\boldsymbol{x}_i\}_{i=1}^n=\{(u_i,v_i,\delta_{1i}, \delta_{2i},\delta_{3i}, \boldsymbol{z}_i)\}_{i=1}^n$ are independent observations of  $(U,V,\Delta_{1}, \Delta_{2},\Delta_{3}, \boldsymbol{Z})$ with sample size $n$. In addition, for the focus of this paper as pointed out in Section~\ref{sec1}, we assume the survival function of $T$ conditional on covariates follows the Cox PH model, that is, $S(t|\boldsymbol{z})=\Pr(T>t|\boldsymbol{Z}=\boldsymbol{z})=e^{-\Lambda(t)e^{\boldsymbol{\beta}^\top \boldsymbol{z}}}$ for $\boldsymbol{\beta}=(\beta_1,\cdots,\beta_K)^\top$ with the baseline cumulative hazard function $\Lambda$ for $T$. Also assume that $T$ is independent of $(U,V)$ given $\boldsymbol{Z}$ and that $(U,V)$ is independent of $\boldsymbol{Z}$. The likelihood of $\{x_i\}_{i=1}^n$ can be derived and reduced to
\begin{equation}\label{sievelike}
L_f(\boldsymbol{\beta},\Lambda;)=\prod_{i=1}^n \left\{1-S(u_i|\boldsymbol{z}_i)\right\}^{\delta_{1i}}\left\{S(u_i|\boldsymbol{z}_i)-S(v_i|\boldsymbol{z}_i)\right\}^{\delta_{2i}}S(v_i|\boldsymbol{z}_i)^{\delta_{3i}}.
\end{equation}

When the nonparametric baseline survival function is smooth (which is a widely adopted assumption in survival analysis~\cite[]{kalbfleisch2002}), it is valid to adopt the sieve spline approach to study the interval-censored Cox PH model.  This paper is based on adopting the monotone I-splines~\cite[]{ramsay} to approximate the cumulative baseline hazard function $\Lambda$. The I-spline form for the cumulative baseline hazard function is given by
\begin{equation}\label{sievelambda}
\Lambda_n(t)=\sum_{j=1}^{p_n}\alpha_jI_j^l(t),
\end{equation}
with nonnegative spline coefficients $\alpha_j$s, where $I_j^l$s are the monotone I-spline basis functions of order $l$ and the number of I-spline basis functions $p_n=O\left(n^{\nu}\right)$ for $0<\nu<1$. Then $L_f(\boldsymbol{\beta},\Lambda_n;)$ becomes the I-spline based sieve likelihood, where $L_f$ is defined by (\ref{sievelike}).

Note that the sieve likelihood we just described is the same as the one adopted by~\cite{wang:hudgens:2016} but differs from that described by~\cite{zhang:Hua:Huang}.
Specifically, in~\cite{zhang:Hua:Huang} $\log \Lambda$ is approximated using B-splines to construct the spline-based likelihood. However, by the equivalence between the B-splines and the I-splines established by~\cite{wu:thesis:2010}, these two sieve approaches are similar in practice.

Assume that $\left(\hat{\boldsymbol{\beta}}_n,\hat{\Lambda}_n\right)$ is the sieve spline-based maximum likelihood estimator (MLE) of the true parameter and baseline cumulative hazard function $(\boldsymbol{\beta}_0,\Lambda_0)$ with $\boldsymbol{\beta}_0=(\beta_{01},\cdots,\beta_{0K})^\top$. Also assume that $\hat{\Lambda}_{n0}$ is the sieve MLE of $\Lambda_0$ when $\boldsymbol{\beta}_0$ is known. Then we construct the likelihood ratio statistic $2\log\frac{L_f\left(\hat{\boldsymbol{\beta}}_n,\hat{\Lambda}_n;\right)}{L_f\left(\boldsymbol{\beta}_0,\hat{\Lambda}_{n0};\right)}$ of the global test for $\boldsymbol{\beta}=\boldsymbol{\beta}_0$. For different local tests of $\boldsymbol{\beta}$, the likelihood ratio statistic can be constructed similarly. For example, regarding the test for $\beta_1=\beta_{01}$, assume $\left(\hat{\boldsymbol{\beta}}_{n,-1},\hat{\Lambda}_{n,-1}\right)$ is the sieve MLE of
$\left({\boldsymbol{\beta}}_{0,-1}, \Lambda_0\right)$ with
$\boldsymbol{\beta}_{0,-1}=(\beta_{02},\cdots,\beta_{0K})^\top$ when $\beta_{01}$ is known, then the corresponding likelihood ratio statistic is given by $2\log\frac{L_f\left(\hat{\boldsymbol{\beta}}_{n},\hat{\Lambda}_{n};\right)}{L_f\left\{\left(\beta_{01},\hat{\boldsymbol{\beta}}_{n,-1}^\top\right)^\top, \hat{\Lambda}_{n,-1};\right\}}$.

The proposed likelihood ratio approach relies on the sieve MLEs, which depend on the spline coefficients that maximize the likelihood given a set of I-spline basis functions. Therefore, before computing the sieve MLEs, it is necessary to construct the I-spline basis functions using a reasonable knot sequence determined by the observed data. Specifically, the number of the interior knots of the knot sequence is chosen to be the largest integer smaller than $n^{\nu}$ with $\nu=1/3$, and the interior knots are positioned at the quantiles of the union of the censoring times $\{(u_i,v_i)\}_{i=1}^n$.
In addition, cubic I-spline basis functions (order 4) are used for the computation based on our experience and the literature.
For further details, see~\cite{wu:thesis:2010}.

 To address the computation issues mentioned in Section~\ref{sec1} of the constrained projection method~\cite[]{zhang:Hua:Huang} and the EM algorithm~\cite[]{wang:hudgens:2016}, we propose to reparameterize the spline coefficients to convert the constrained maximization problem to an unconstrained one. Specifically, in the likelihood $L_f(\beta,\Lambda_n;)$, let $\Lambda_n(t)=\sum_{j=1}^{p_n}e^{\alpha_j}I_j^l(t)$. Under this new formulation, the nonnegative constraints enforced for the spline coefficients in (\ref{sievelambda}) are no longer necessary since the reparameterized coefficients $e^{\alpha_j}$s are always nonnegative.
For the resulting unconstrained maximization problem, the sieve MLE can be computed using the Newton method or any other unconstrained optimization algorithm.

\section{Asymptotic distribution of  likelihood ratio statistic}
\label{sec3}
In this section, we derive the asymptotic distribution of the proposed likelihood ratio statistic based on the semiparametric sieve spline-based MLE. The forthcoming Theorem~\ref{thm1} establishes that, under the null hypothesis $H_0: \boldsymbol{\beta} = \boldsymbol{\beta}_0$, the likelihood ratio statistic converges weakly to a $\chi^2$ distribution. This result provides the theoretical foundation for the global test of all regression coefficients. Parallel theory can also be developed for local tests; for example, testing the first regression coefficient under the null hypothesis $H_0: \beta_1 = \beta_{01}$.


\begin{thm}\label{thm1}
Let
\[
lr(\boldsymbol{\beta}_0)
= 2\log \frac{L_f\!\left(\hat{\boldsymbol{\beta}}_n,\hat{\Lambda}_n;\right)}
{L_f\!\left(\boldsymbol{\beta}_0,\hat{\Lambda}_{n0};\right)}.
\]
Under the null hypothesis $H_0:\boldsymbol{\beta}=\boldsymbol{\beta}_0$ and conditions C1--C4 in the Supplementary Document, as $n \to \infty$,
\[
lr(\boldsymbol{\beta}_0) \xrightarrow{d} \chi^2_K.
\]
\end{thm}

It is natural to attempt a proof of Theorem~\ref{thm1} by adapting the approach used to develop Theorem 3.1 in~\cite{murphy:vandervaart:1997}.
However, a unique challenge for this spline-based testing approach is that the proposed likelihood ratio statistic depends on the sieve MLE, which is restricted to the monotone spline function class rather than the class of arbitrary monotone functions. As a result, key arguments in \cite{murphy:vandervaart:1997} are not directly applicable to our setting. Our proof in the Supplementary Document relies on novel techniques that exploit properties of B-spline basis functions, as discussed in \cite{deboor:01} and \cite{schumaker:2007}.

\section{Simulation Studies}
\label{sec4}

The main objective of our simulation is to compare our proposed likelihood ratio approach with the corresponding Wald-type approaches based on  methods for estimating the information matrix adopted by~\cite{zhang:Hua:Huang} and ~\cite{wang:hudgens:2016}, respectively.

The interval-censored data we used for the simulation is generated similarly to~\cite{zhang:Hua:Huang}. Specifically, let the event time $T$ follow the Cox PH model $\Pr(T>t|\boldsymbol{Z}=\boldsymbol{z})=e^{-\Lambda(t)e^{\boldsymbol{\beta}_0^\top \boldsymbol{z}}}$, with $\boldsymbol{\beta}_0=(\beta_{01}, \beta_{02}, \beta_{03})^\top=(-1.0, 0.5, 1.5)^\top$, $\Lambda_0(t)=t^{1/2}$, and $\boldsymbol{Z}=(Z_1,Z_2,Z_3)^\top$, where $Z_1$ follows the uniform distribution on $[0,1]$, $Z_2$ follows the standard normal distribution and $Z_3$ follows the Bernoulli distribution with $\Pr(Z_3=1)=0.5$.  And let two independent censoring times $U_0$ and $V_0$ both follow the exponential distribution with its mean equal to $2$. For one observation  $(u_0, v_0, t, \boldsymbol{z})$ of $(U_0,V_0, T,\boldsymbol{Z})$, first let $u=\min(u_0,v_0)$ and $v=\max(u_0,v_0)$; then let $u=\min(u,5)$ and $v=\min(v,5)$. Next if $v-u\ge0.05$, stop; if $v-u<0.05$ and $v>0.05$, let $u=v-0.05$; if $v-u<0.05$ and $v<0.05$, let $v=v+0.05$. By this approach, one interval-censored observation is obtained as $(u,v,\delta_{1},\delta_{2},\delta_{3},\boldsymbol{z})$, where $\delta_1=1_{[t\le u]}$, $\delta_2=1_{[u<t\le v]}$, $\delta_3=1-\delta_1-\delta_2$. Finally the interval-censored data sample for the simulation is made of this type of interval-censored observations. And these interval-censored observations are generated independently.


First, we assess our proposed unconstrained maximization (UM) approach by comparing its performance to that of the EM algorithm adopted by~\cite{wang:hudgens:2016}. The UM approach is not only significantly faster at computing, but producing slightly less biased sieve MLE. This comparison is based on 100 independent samples of the interval-censored data described in the previous paragraph, with sample sizes of 200 and 500. The results, including computing time and estimation bias, are presented in Table~\ref{tab1} and Figure~\ref{fig1}. For small sample sizes (e.g., 50), the function implementing the EM algorithm~\cite[]{wang:hudgens:2016} in the R package \texttt{ICsurv}~\cite[]{icsurv} frequently fails to run successfully. Moreover, as noted by \cite{wang:hudgens:2016} and verified in our preliminary study, the constrained projection algorithm used by \cite{zhang:Hua:Huang} is also relatively slow and prone to failure. By contrast, the proposed UM approach is robust, with no failures occurring for all sample sizes considered. Hence, the proposed UM approach is superior at computing the sieve MLE and is adopted for the remainder of the simulation studies.

\begin{table}[htp]
\centering\caption{Comparisons between UM and EM with $\boldsymbol{\beta}_0=(-1.0, 0.5, 1.5)^\top$}\label{tab1}
\begin{tabular}{rrrrr}
\hline
& \multicolumn{3}{c}{Bias}&\multirow{2}{*}{Time (seconds)}\\
\cmidrule(lr){2-4}
&$\beta_{01}$ &$\beta_{02}$ & $\beta_{03}$& \\
\hline
Sample size 200&&&&\\
UM&-7.06e-2&3.87e-2& 9.56e-2&4.55\\
EM &-9.47e-2&5.07e-2&1.21e-1&326.00 \\
\hline
Sample size 500&&&&\\
UM&-8.74e-3&1.10e-2& 2.36e-4&11.36\\
EM& -2.69e-2&1.97e-2&1.96e-2& 780.04 \\
\hline
\end{tabular}
\end{table}

\begin{figure}
\centering
\vspace{-0ex}
\mbox{\hspace{0ex}
\subfigure{\scalebox{0.45}{\includegraphics{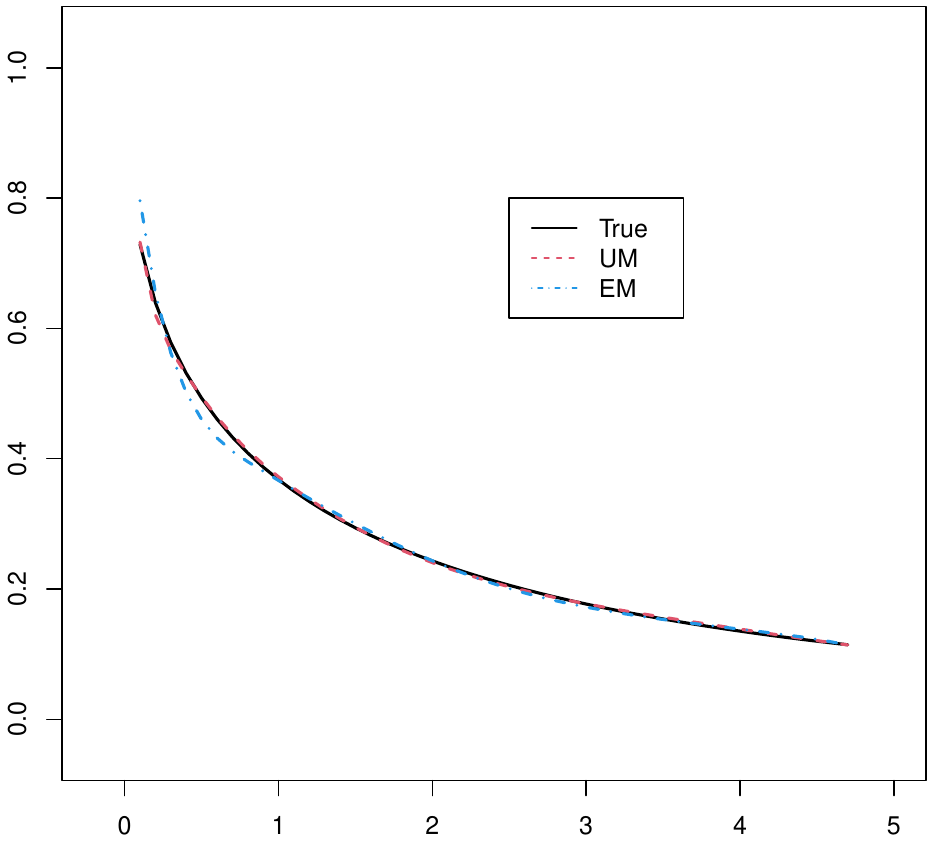}}} 
\quad
\subfigure{\scalebox{0.45}{\includegraphics{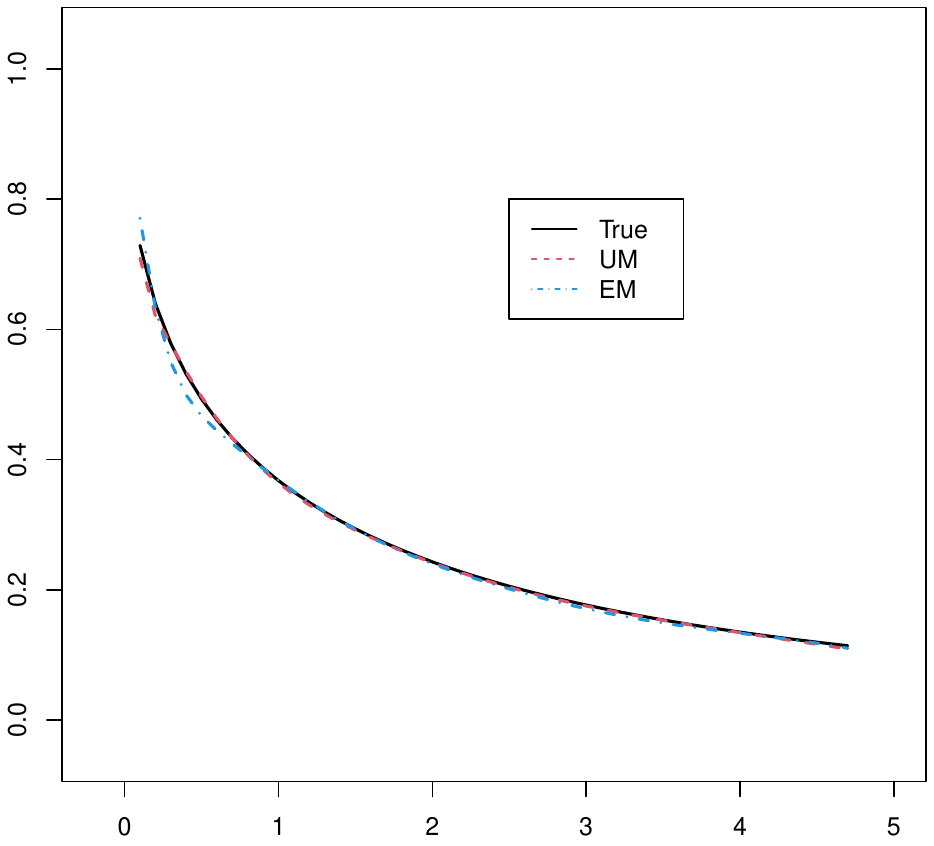}}}
}
\caption{Sieve MLE of baseline survival function, $S(t)=e^{-t^{1/2}}$, by UM and EM, with sample size 200 (left) and 500 (right)}\label{fig1}
\end{figure}

Tables~\ref{tab2}-\ref{tab3} present the estimated power and type I error, respectively. Specifically, they report the rejection rates (RR) among 10000 independent experiments for three local tests for $\beta_{1}=0$, $\beta_{2}=0$ and $\beta_{3}=0$, respectively, and for the global test for $\boldsymbol{\beta}=\boldsymbol{0}$, when conducting our proposed likelihood ratio test (LRT),  the Wald-type test based on~\cite{zhang:Hua:Huang} (WaldT1), and  the Wald-type test based on~\cite{wang:hudgens:2016} (WaldT2). Sample sizes 50, 200 and 200 are chosen regarding both the power and the type I error, for which all tests are conducted at a two-sided level of 0.05. The power estimation is based on the aforementioned interval-censored data. For Type I error, the interval-censored data is generated by the same approach except that the event time $T$ follows the Cox PH model with ${\boldsymbol{\beta}}_0=\boldsymbol{0}$ instead of $\boldsymbol{\beta}_0=(-1.0, 0.5, 1.5)^\top$. 

Table~\ref{tab2} shows that the LRT consistently achieves the highest power among the three testing approaches across all sample sizes. Table~\ref{tab3} indicates that the estimated type I error converges toward the nominal 5\% level as the sample size increases from 200 to 500. For sample size 50, although the LRT yields slightly inflated type I error rates, the WaldT1 and WaldT2 exhibit much larger deviations from the 5\% level, which are likely caused by the biased variance estimates of $\hat{\boldsymbol{\beta}}_n$ produced separately by~\cite{zhang:Hua:Huang} and~\cite{wang:hudgens:2016}.


Lastly, the widely used partial likelihood approach is assessed for being an alternative for analyzing the interval-censored data. To facilitate this, the same interval-censored data is converted to a right-censored one through midpoint imputation, by assigning the midpoint of each finite censoring interval as the event time. The maximum partial likelihood estimation (MPLE) involving the data conversion fails to converge, since the estimation bias for ${\boldsymbol{\beta}}_0=(-1.0, 0.5, 1.5)^\top$ is $(0.33, -0.17, -0.37)^\top$ for sample size 200, but becomes $(0.35, -0.18, -0.38)^\top$ (even larger) for sample size 500. These results are based on 1000 independent samples. This finding agrees with the corresponding arguments made by~\cite{wu:halabi:2019}.


\begin{table}[htp]
\centering\caption{Rejection rates of local and global tests when $\boldsymbol{\beta}_0=(-1.0, 0.5, 1.5)^\top$}\label{tab2}
\begin{tabular}{rrrrr}
\hline
&\multicolumn{4}{c}{Rejection rate}\\
\cmidrule(lr){2-5}
&$H_0:\beta_{1}=0$&$H_0:\beta_{2}=0$&$H_0:\beta_{3}=0$&$H_0:\boldsymbol{\beta}=\boldsymbol{0}$\\
\hline
Sample size 50&&&&\\
WaldT1   &8.77\%&26.49\%&70.13\%& 38.99\%\\
WaldT2  &47.97\%&61.10\%&92.95\%& 65.60\%\\
LRT  &34.57\%&69.27\%&97.57\%& 98.56\%\\
\hline
Sample size 200&&&&\\
WaldT1  &77.71\%&99.69\%&100\%&100\%\\
WaldT2  &78.39\%&99.80\%&99.99\%&99.99\%\\
LRT  &84.09\%&99.86\%&100\%&100\%\\
\hline
Sample size 500&&&&\\
WaldT1 &99.65\%&100\%&100\%&100\%\\
WaldT2   &99.64\%&100\%&100\%&100\%\\
LRT  &99.77\%&100\%&100\%&100\%\\
\hline
\end{tabular}
\end{table}

\begin{table}[htp]
\centering\caption{Rejection rates of local and global tests when $\boldsymbol{\beta}_0=(0, 0, 0)^\top$}\label{tab3}
\begin{tabular}{rrrrr}
\hline
&\multicolumn{4}{c}{Rejection rate}\\
\cmidrule(lr){2-5}
&$H_0:\beta_{1}=0$&$H_0:\beta_{2}=0$&$H_0:\beta_{3}=0$&$H_0:\boldsymbol{\beta}=\boldsymbol{0}$\\
\hline
Sample size 50&&&&\\
WaldT1  &2.07\%&1.69\%&2.21\%& 0.83\%\\
WaldT2  &24.39\%&8.41\%&14.99\%& 1.74\%\\
LRT &7.08\%&6.71\%&7.00\%& 7.39\%\\
\hline
Sample size 200&&&&\\
WaldT1  &4.27\%&4.28\%&4.17\%&3.55\%\\
WaldT2  &4.49\%&5.45\%&4.81\%&4.22\%\\
LRT &5.48\%&5.68\%&5.19\%&5.34\%\\
\hline
Sample size 500&&&&\\
WaldT1  &4.57\%&4.44\%&4.54\%&4.45\%\\
WaldT2  &4.60\%&4.74\%&4.86\%&4.61\%\\
LRT &5.22\%&4.84\%&5.06\%&4.93\%\\
\hline
\end{tabular}
\end{table}

\section{A Real Example}
\label{sec5}

The data were collected for a prospective study to assess the HIV-1
infection rate among people with hemophilia. The individuals were at risk of contracting HIV-1 through plasma from blood donors. In this study, 544 patients were categorized into four groups based on their average annual dose of blood products. Specifically, they were classified into High, Medium, Low, or No dose group. The goal of the
study was to compare the HIV-1 infection rate between these exposure groups. The exact time to HIV-1 infection for each patient was unknown, but was known to occur before, between, or after
certain observation times.
The Hemophilia dataset is available in the R package ICsurv~\cite[]{icsurv}. and more details pertaining to the data can be found in~\cite{goedert:kessler:1989} and \cite{kroner:rosenberg:1994}.

Assume the Hemophilia data follows the Cox PH model, we use the three testing approaches, the WaldT1, WaldT2 and LRT, to assess the effect of blood product dose levels on the risk of HIV-1 infection. Table~\ref{tab4} presents the p-values of three local tests for the effect of each dose level (Low, Medium, High) on the risk of contracting HIV-1 and the global test for the effect of any dose level. In general, the results of the local tests indicate that receiving blood products significantly increased the risk of contracting HIV-1 for study subjects. More specifically, the p-values for the WaldT1 and LRT imply that the test for the effect of a higher dose level yields a more significant result and the global test gives the most significant result, which reflect the practical situation. In addition, the LRT is consistently more significant than the WaldT1. On the other hand, although the WaldT2 shows highly significant p-values for all three local tests, these p-values fail to reflect the trend of the three dose levels. More concerning the WaldT2 results in a p-value of 1 for the global test, which is clearly inappropriate and raises doubts about its practical utility. Figure~\ref{fig2} displays nearly identical estimates of the baseline survival function obtained via the EM and UM computation approaches. This indicates that, among patients who did not receive blood products, fewer than 15\% contracted HIV-1.

\begin{table}[htp]
\centering\caption{P-values for blood product dose effect on HIV-1 risk (Hemophilia data)}\label{tab4}
\begin{tabular}{rrrrr}
\hline
&\multicolumn{4}{c}{P-value}\\
\cmidrule(lr){2-5}
&Low &  Medium & High & Any\\
\hline
WaldT1 &7.73e-15&3.32e-31&5.23e-41& 7.40e-43\\
WaldT2 &4.06e-88&1.53e-158&2.52e-31& 1\\
LRT &4.00e-19&1.56e-51&2.90e-52& 3.40e-76\\
\hline
\end{tabular}
\end{table}

\begin{figure}
\begin{center}
			\scalebox{0.5}{\includegraphics[angle =0]{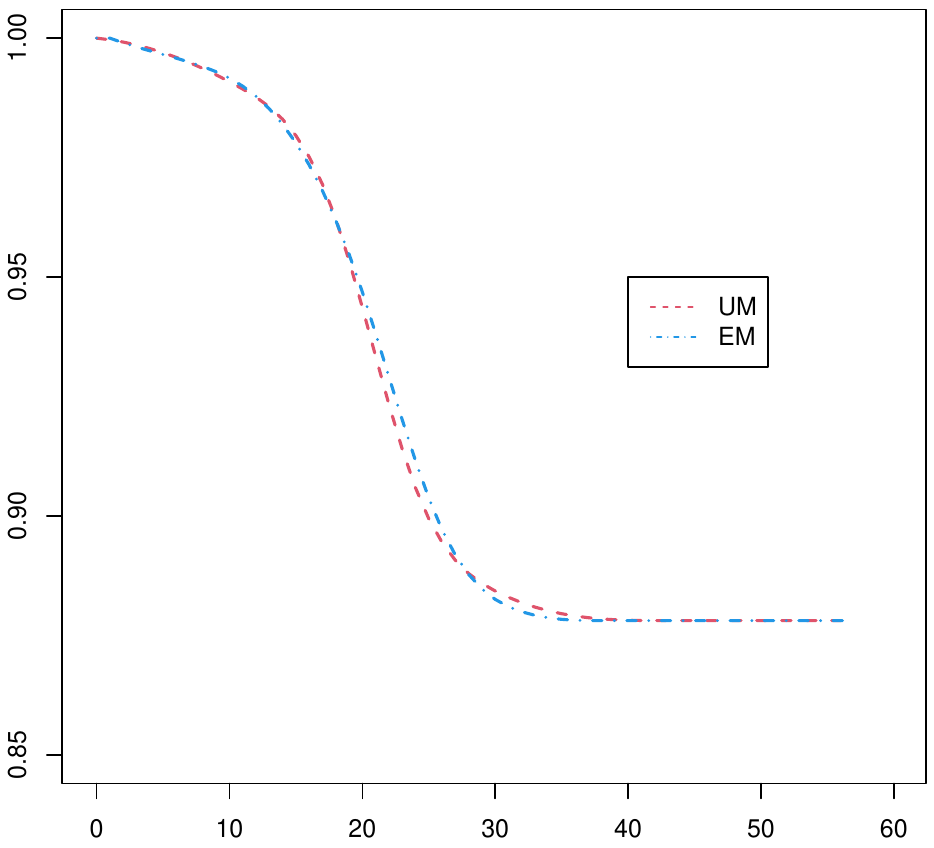}}
\caption{Estimated baseline survival functions by UM and EM}\label{fig2}
\end{center}
\end{figure}

\section{Concluding Remarks}
\label{sec6}

In this paper, we develop a framework of the likelihood ratio test (LRT) for the Cox PH model with interval-censored data, a common feature in clinical studies where event times are assessed intermittently. Conventional approaches that approximate interval-censored data as right-censored can introduce bias and misleading inference; our method addresses these challenges with both computational efficiency and theoretical rigor.

We contribute a robust unconstrained maximization (UM) algorithm to compute the sieve semiparametric MLE. By reparameterizing spline coefficients, the UM algorithm avoids the instability and slow convergence of constrained or EM-based methods and can be applied to other nonparametric and semiparametric models. We also establish the asymptotic chi-square distribution of the LRT statistic, leveraging new results on B-spline–based estimators, thereby extending likelihood-based inference tools for interval-censored data.

Simulation studies demonstrate that the LRT outperforms existing Wald-type tests (WaldT1 and WaldT2). WaldT1 is overly conservative in small samples, while WaldT2 exhibits erratic type I error rates. By contrast, the LRT maintains stability, achieves higher power, and uniquely ensures that the global test is more powerful than its corresponding local tests, thereby providing a reliable inferential framework in which the global test is followed by post-hoc analyses, even though it shows slightly inflated type I error rates in very small samples.
Analysis of the Hemophilia dataset illustrates the practical utility of the LRT, detecting dose-dependent increases in HIV-1 infection risk, whereas WaldT2 produces unreliable results. These findings highlight the robustness and interpretability of our approach in real-world applications.

The proposed LRT offers a reliable and flexible framework for analyzing interval-censored survival data, with particular advantages in small-sample studies. It can extend naturally to more complex semiparametric survival models, including frailty, competing risks, and multi-state designs. Future work could explore variable selection strategies for high-dimensional covariates to further enhance performance.

In summary, our method provides both methodological innovation and practical utility, equipping researchers with a robust tool for interval-censored time-to-event data in clinical and epidemiological studies.

\section*{Funding}

This research was supported in part by the National Cancer Institute Awards R21 CA263950, R01 CA256157 and R01 CA249279, the  Food and Drug Administration Award 1U01FD007857-01, the United States Army Medical Research Award HT9425-23-1-0393, and the Prostate Cancer Foundation.

\bibliographystyle{natbib}
\bibliography{paper_ref}

\section*{Supplementary Material}
The Supplementary Material provides the technical details required for the theoretical developments.
 \subsection*{Technical assumptions}
\begin{enumerate}[C1]
\item
 The true hazard function $\lambda_0$ has a positive lower bound on $[0,\tau_1]$ and   has up to $(p-1)$th continuous derivatives on $[0,\tau_1]$. In addition, its $p$th derivative satisfy $\left|\lambda_0^p(v)-\lambda_0^p(u)\right|\le c |v-u|^r$ for $u,v$ with $0\le u<v\le\tau_1$ and a constant $c>0$, where $p$ is a natural number and $r>0$ satisfying $p+r>3/2$.
\item
Let $[\tau_0,\tau_1]$ be the support interval for the union of $U$ and $V$ with $0<\tau_0<\tau_1$. The joint density of $(U,V)$ is smooth enough and has a positive lower bound in their closed joint support triangle obtained by $V-U\ge \eta$ for $0<\eta < \tau_1-\tau_0$.
\item
Covariate variable $\boldsymbol{Z}$ is bounded, that is, there exists a positive scalar $z_0$ such that $\|\boldsymbol{Z}\|<z_0$. Here $\|\cdot\|$ denotes Euclidean norm.  
\item
There exists $\eta\in(0,1)$, such that $\boldsymbol{a}^\top Var(\boldsymbol{Z})\boldsymbol{a}\le \eta \boldsymbol{a}^\top E(\boldsymbol{Z}\boldsymbol{Z}^\top)\boldsymbol{a}$ for all $\boldsymbol{a}\in \mathbb{R}^K$.
\end{enumerate}

\subsection*{Asymptotic consistency of semiparametric sieve MLE}
\label{asymptotic}
 To establish Theorem~1 in the main paper, it is necessary to first show the asymptotic convergence and normality of the proposed semiparametric sieve MLE.
  This section describes these asymptotic results for the proposed sieve MLE.
  The technical details for developing these results closely follow~\cite{wu:xu:19} and are therefore omitted.

  First, because a monotone I-spline function can be expressed in B-spline form~\cite[]{wu:thesis:2010}, we replace the I-spline-based baseline cumulative hazard function described in Section~2 in the main paper with an equivalent monotone B-spline function for theoretical arguments. This allows us to leverage the well-developed properties of B-spline functions in~\cite{deboor:01} and~\cite{schumaker:2007}. Accordingly, under conditions C1--C4, the B-spline based sieve MLEs $\left(\hat{\boldsymbol{\beta}}_n,\hat{\Lambda}_n\right)$ and $\hat{\Lambda}_{n0}$ satisfy $d\left\{\left(\hat{\boldsymbol{\beta}}_n,\hat{\Lambda}_n\right), (\boldsymbol{\beta}_0,\Lambda_0)\right\}=O_P\left(n^{-\frac{p+r}{2(p+r)+1}}\right)$ and $d\left(\hat{\Lambda}_{n0},\Lambda_0\right)=O_P\left(n^{-\frac{p+r}{2(p+r)+1}}\right)$, respectively, where the metric $d(,)$ is defined as in~\cite{zhang:Hua:Huang}.


Regarding the asymptotic normality of the parametric component, the efficient score for the covariate coefficients $\boldsymbol{\beta}$ is necessary, which is based on the theory of semiparametric efficiency~\cite[]{bickel:93}.
Let the likelihood function for a single observation be $$L(\boldsymbol{\beta},\Lambda;\boldsymbol{x})=\left\{1-e^{-\Lambda(u)e^{\boldsymbol{\beta}^\top \boldsymbol{z}}}\right\}^{\delta_1}
\left\{e^{-\Lambda(u)e^{\boldsymbol{\beta}^\top\boldsymbol{z}}} - e^{-\Lambda(v)e^{\boldsymbol{\beta}^\top\boldsymbol{z}}}\right\}^{\delta_2}
\left\{e^{-\Lambda(v)e^{\boldsymbol{\beta}^\top\boldsymbol{z}}}\right\}^{\delta_3},$$
where $\boldsymbol{x}=(u,v,\delta_1,\delta_2,\delta_3,\boldsymbol{z})$, with $\boldsymbol{z}$ being the covariates, $\delta_1=1_{[t\le u]}$, $\delta_2=1_{[u< t\le u]}$ and $\delta_3=1-\delta_1-\delta_2$, for the event time $t$ and its censoring times $u$ and $v$. Let $l(\boldsymbol{\beta},\Lambda;\boldsymbol{x})=\log L(\boldsymbol{\beta},\Lambda;\boldsymbol{x})$. Consider a parametric smooth submodel with parameter $\left(\boldsymbol{\beta}, \Lambda_{(h,s)}\right)$, where
$\Lambda_{(h,s)}=\Lambda+sh$. Then
$\Lambda_{(h,0)}=\Lambda$ and $\left.\frac{\partial \Lambda_{(h,s)}}{\partial s}\right\vert_{s=0}=h$.  The score operator for $\Lambda$ is
\begin{equation}\label{score_op}
l_\Lambda\left(\boldsymbol{\beta},\Lambda;\right)h=\left.\frac{\partial}{\partial s}l\left(\boldsymbol{\beta},\Lambda_{(h,s)};\right)\right\vert_{s=0}
\end{equation}
According to~\citep{bickel:93}, the efficient score vector for $\boldsymbol{\beta}_0$ with dimension $K$ is $l_{\boldsymbol{\beta}}(\boldsymbol{\beta}_0,\Lambda_0;)-l_\Lambda(\boldsymbol{\beta}_0,\Lambda_0;)\boldsymbol{h}_\ast$, where
\begin{equation}\label{hstar}
\boldsymbol{h}_\ast=(h_{\ast 1},\cdots,h_{\ast K})^\top=
\arg\min_{\boldsymbol{h}}E\left\|l_{\boldsymbol{\beta}}(\boldsymbol{\beta}_0,\Lambda_0;)-l_\Lambda(\boldsymbol{\beta}_0,\Lambda_0;)\boldsymbol{h}\right\|^2,
\end{equation}
with $\boldsymbol{h}=(h_1,\cdots,h_K)^\top$ and $l_\Lambda(\boldsymbol{\beta}_0,\Lambda_0;)\boldsymbol{h}=
(l_\Lambda(\boldsymbol{\beta}_0,\Lambda_0;)h_1,\cdots,
l_\Lambda(\boldsymbol{\beta}_0,\Lambda_0;)h_K)^\top$. Note that $\boldsymbol{h}_\ast$ is called the least favorable direction and can be solved for each component separately. 

Denote the efficient score for $\boldsymbol{\beta}_0$ by $\tilde{l}(\boldsymbol{\beta}_0,\Lambda_0;)$, where
\begin{equation}\label{efficientscore}
\tilde{l}(\boldsymbol{\beta}_0,\Lambda_0;)=l_{\boldsymbol{\beta}}(\boldsymbol{\beta}_0,\Lambda_0;)-l_\Lambda(\boldsymbol{\beta}_0,\Lambda_0;)\boldsymbol{h}_\ast.
\end{equation}
It follows that under conditions C1--C4, the asymptotic normality holds, that is,
\begin{equation}\label{normality}
\sqrt{n}\left(\hat{\boldsymbol{\beta}}_n-\boldsymbol{\beta}_0\right)
=\sqrt{n}I^{-1}(\boldsymbol{\beta}_0)\mathbb{P}_n\tilde{l}(\boldsymbol{\beta}_0,\Lambda_0;)+o_p(1)\to_d N\left(0,I^{-1}(\boldsymbol{\beta}_0)\right)
\end{equation}
where $I(\boldsymbol{\beta}_0)=E\tilde{l}(\boldsymbol{\beta}_0,\Lambda_0;)\tilde{l}(\boldsymbol{\beta}_0,\Lambda_0;)^\top$.

\subsection*{Proof of Theorem~1 in the main paper}


 Throughout the remainder of this document, $\|\cdot\|_\infty$ denotes the supremum norm on $[\tau_0, \tau_1]$, unless stated otherwise.
 
By the Jackson-type theorem in~\cite{deboor:01}, there exists a sequence of B-spline vectors $\{\boldsymbol{h}_n\}_{n>N_0}$ with $\boldsymbol{h}_n=(h_{n1},\ldots,h_{nK})^\top$ such that, for each $k\in\{1,\ldots,K\}$,
\[
\|h_{nk}-h_{\ast k}\|_\infty = O(n^{-\nu}), \qquad 0<\nu<1,
\]
where $h_{\ast k}$ is a component determined by (\ref{hstar}), and the derivatives of $\{h_{nk}\}_{n>N_0}$ are uniformly bounded. Here $N_0$ is sufficiently large, and $O(n^\nu)$ refers to the number of subintervals determined by the interior knots, which are placed uniformly or nearly uniformly.


In what follows, we repeatedly use the minimum and maximum of the true hazard $\lambda_0$ on $[0,\tau_1]$, denoted by $b_l$ and $b_u$, respectively. Under condition C1, these quantities are well defined and satisfy $0<b_l<b_u$.

Let $\mathcal{F}_n$ denote the class of monotone B-spline functions on $[\tau_0,\tau_1]$ whose values at $\tau_0$ are bounded above by $M\tau_0$ with $M=2b_u$, and whose derivatives are uniformly bounded by $M$ on $[\tau_0,\tau_1]$. In addition, all functions in $\mathcal{F}_n$ and $\boldsymbol{h}_n$ are constructed from the same B-spline basis. Finally, let $\nu=\tfrac{1}{2(p+r)+1}$ for $O(n^\nu)$.

We assume that the computation of both $\hat{\Lambda}_n$ and $\hat{\Lambda}_{n0}$ is restricted to $\mathcal{F}_n$, which facilitates the derivation of consistency and asymptotic normality for the sieve MLE (see the previous section) and is essential for the remaining proof.

Define $\boldsymbol{\omega}_{\boldsymbol{t}}(\boldsymbol{\beta}, \Lambda, \boldsymbol{g})=\left({\boldsymbol{t}}, \Lambda_{\boldsymbol{t}}(\boldsymbol{\beta},\Lambda,\boldsymbol{g})\right)$, where
$$\Lambda_{\boldsymbol{t}}(\boldsymbol{\beta},\Lambda,\boldsymbol{g})=\Lambda + (\boldsymbol{\beta} - {\boldsymbol{t}})^\top \boldsymbol{g},$$ for vector function $\boldsymbol{g}$ with $K$ components.

Then $\boldsymbol{\omega}_{\hat{\boldsymbol{\beta}}_n}\left(\hat{\boldsymbol{\beta}}_n, \hat{\Lambda}_n,\boldsymbol{h}_n\right)=\left(\hat{\boldsymbol{\beta}}_n, \hat{\Lambda}_n\right)$ and $\boldsymbol{\omega}_{\boldsymbol{\beta}_0}\left(\boldsymbol{\beta}_0, \hat{\Lambda}_{n0}, \boldsymbol{h}_n\right)=\left(\boldsymbol{\beta}_0,\hat{\Lambda}_{n0}\right)$.  Denote
\begin{align}\label{derilambda}
\hat{\lambda}_n(y)\equiv\frac{\partial\hat{\Lambda}_{n}(y)}{\partial y}~~\text{and}~~\hat{\lambda}_{n0}(y)\equiv\frac{\partial\hat{\Lambda}_{n0}(y)}{\partial y}.
\end{align}
If $\hat{\lambda}_{n}>b_1$  and $\hat{\lambda}_{n0}>b_1$ for $b_1=b_l/2$, then there exists a small $\epsilon_0>0$, we have
\begin{align}\label{lambdan}
\epsilon_0\sum_{k=1}^K\left|h_{nk}(v)-h_{nk}(u)\right|<
\hat{\Lambda}_n(v)-\hat{\Lambda}_n(u)
\end{align}
and
\begin{align}\label{lambdan0}
\epsilon_0\sum_{k=1}^K\left|h_{nk}(v)-h_{nk}(u)\right|<
\hat{\Lambda}_{n0}(v)-\hat{\Lambda}_{n0}(u),
\end{align}
for any $u$, $v$ with $\tau_0\le u<v\le \tau_1$. (\ref{lambdan}) and (\ref{lambdan0}) can be both easily proved from the fact that the derivatives of  $\left\{h_{nk}\right\}_n$ are uniformly bounded, when $n$s are large.
Then for $\|\boldsymbol{\beta}-{\boldsymbol{t}}\|<\epsilon_0$ (with Euclidean norm $\|\cdot\|$),
\begin{align*}
\Lambda_{\boldsymbol{t}}\left(\boldsymbol{\beta},\hat{\Lambda}_n,\boldsymbol{h}_n\right)(v)-
\Lambda_{\boldsymbol{t}}\left(\boldsymbol{\beta},\hat{\Lambda}_n,\boldsymbol{h}_n\right)(u)>0
\end{align*}
and
\begin{align*}
\Lambda_{\boldsymbol{t}}\left(\boldsymbol{\beta},\hat{\Lambda}_n,\boldsymbol{h}_n\right)(v)-
\Lambda_{\boldsymbol{t}}\left(\boldsymbol{\beta},\hat{\Lambda}_n,\boldsymbol{h}_n\right)(u)>0,
\end{align*}
for any $u$, $v$ with $\tau_0\le u<v\le \tau_1$,
by (\ref{lambdan}) and (\ref{lambdan0}), respectively.

 Denote $\lambda_{\boldsymbol{t}}\left(\boldsymbol{\beta},\hat{\Lambda}_n,\boldsymbol{h}_n\right)(y)\equiv\frac{\partial\Lambda_{\boldsymbol{t}}\left(\boldsymbol{\beta},\hat{\Lambda}_n,\boldsymbol{h}_n\right)(y)}{\partial y}$ and $\lambda_{\boldsymbol{t}}\left(\boldsymbol{\beta},\hat{\Lambda}_{n0},\boldsymbol{h}_n\right)(y)\equiv\frac{\partial\Lambda_{\boldsymbol{t}}\left(\boldsymbol{\beta},\hat{\Lambda}_{n0},\boldsymbol{h}_n\right)(y)}{\partial y}$. If
$\hat{\lambda}_{n}<b_2$  and $\hat{\lambda}_{n0}<b_2$ for $b_2=(3b_u)/2$.  Then by the derivatives of $\{h_{nk}\}$ are uniformly bounded, for $\|\boldsymbol{\beta}-{\boldsymbol{t}}\|<\epsilon_0$ with $\epsilon_0$ sufficiently small, we have
\begin{align*}
\lambda_{\boldsymbol{t}}\left(\boldsymbol{\beta},\hat{\Lambda}_n,\boldsymbol{h}_n\right)< M~~\text{and}~~
\lambda_{\boldsymbol{t}}\left(\boldsymbol{\beta},\hat{\Lambda}_{n0},\boldsymbol{h}_n\right)< M,
\end{align*}
on $[\tau_0,\tau_1]$.

Besides, if $\hat{\Lambda}_n(\tau_0)<(2M/3)\tau_0$ and $\hat{\Lambda}_{n0}(\tau_0)<(2M/3)\tau_0$. Then since $\{h_{nk}\}$ are uniformly bounded for sufficiently large $n$s, for $\|\boldsymbol{\beta}-{\boldsymbol{t}}\|<\epsilon_0$ with $\epsilon_0$ sufficiently small, we have
\begin{align*}
\Lambda_{\boldsymbol{t}}\left(\boldsymbol{\beta},\hat{\Lambda}_n,\boldsymbol{h}_n\right)(\tau_0)\le M\tau_0~~\text{and}~~
\Lambda_{\boldsymbol{t}}\left(\boldsymbol{\beta},\hat{\Lambda}_{n0},\boldsymbol{h}_n\right)(\tau_0)\le M\tau_0.
\end{align*}
  The preceding inequalities involving $\Lambda_{\boldsymbol{t}}$s and $\lambda_{\boldsymbol{t}}$s tell that,  $\Lambda_{\boldsymbol{t}}\left(\boldsymbol{\beta},\hat{\Lambda}_n,\boldsymbol{h}_n\right)$ and $\Lambda_{\boldsymbol{t}}\left(\boldsymbol{\beta},\hat{\Lambda}_{n0},\boldsymbol{h}_n\right)$ both belong to $\mathcal{F}_n$. Hence, with both $\hat{\Lambda}_n$ and $\hat{\Lambda}_{n0}$ estimated within $\mathcal{F}_n$, it is true that when $\left\|\hat{\boldsymbol{\beta}}_n-\boldsymbol{\beta}_0\right\|<\epsilon_0$, $b_1<\hat{\lambda}_{n}<b_2$, $b_1<\hat{\lambda}_{n0}<b_2$, $\hat{\Lambda}_n(\tau_0)<(2M/3)\tau_0$ and $\hat{\Lambda}_{n0}(\tau_0)<(2M/3)\tau_0$, we have
\begin{equation}\label{fullinequality}
\mathbb{P}_nl\left(\boldsymbol{\omega}_{\hat{\boldsymbol{\beta}}_n}\left(\boldsymbol{\beta}_0, \hat{\Lambda}_{n0},\boldsymbol{h}_n\right);\right)=
\mathbb{P}_nl\left(\hat{\boldsymbol{\beta}}_n, \Lambda_{\hat{\boldsymbol{\beta}}_n}\left(\boldsymbol{\beta}_0,\hat{\Lambda}_{n0},\boldsymbol{h}_n\right);\right)
\le\mathbb{P}_nl\left(\hat{\boldsymbol{\beta}}_n,\hat{\Lambda}_n;\right)
\end{equation}
and
\begin{equation}\label{partinequality}
\mathbb{P}_nl\left(\boldsymbol{\omega}_{\boldsymbol{\beta}_0}\left(\hat{\boldsymbol{\beta}}_n, \hat{\Lambda}_n,\boldsymbol{h}_n\right);\right)=
\mathbb{P}_nl\left(\boldsymbol{\beta}_0,\Lambda_{\boldsymbol{\beta}_0}\left(\hat{\boldsymbol{\beta}}_n, \hat{\Lambda}_n,\boldsymbol{h}_n\right);\right)
\le\mathbb{P}_nl\left(\boldsymbol{\beta}_0,\hat{\Lambda}_{n0};\right).
\end{equation}
Then by (\ref{partinequality})
\begin{align*}
lr(\boldsymbol{\beta}_0)
&=2\log\frac{L_f\left(\hat{\boldsymbol{\beta}}_n,\hat{\Lambda}_n;\right)}{L_f\left(\boldsymbol{\beta}_0,\hat{\Lambda}_{n0};\right)}\\
&=2n\mathbb{P}_n\left\{l\left(\hat{\boldsymbol{\beta}}_n,\hat{\Lambda}_n;\right)-l\left(\boldsymbol{\beta}_0,\hat{\Lambda}_{n0};\right)\right\}\\
&\le2n\mathbb{P}_n\left\{l\left(\boldsymbol{\omega}_{\hat{\boldsymbol{\beta}}_n}\left(\hat{\boldsymbol{\beta}}_n, \hat{\Lambda}_n,\boldsymbol{h}_n\right);\right)-l\left(\boldsymbol{\omega}_{\boldsymbol{\beta}_0}\left(\hat{\boldsymbol{\beta}}_n, \hat{\Lambda}_n,\boldsymbol{h}_n\right);\right)\right\}\\
&=2n\mathbb{P}_n\left\{-\left.\frac{\partial l\left(\boldsymbol{\omega}_{\boldsymbol{t}}\left(\hat{\boldsymbol{\beta}}_n,\hat{\Lambda}_n,\boldsymbol{h}_n\right);\right)}{\partial {\boldsymbol{t}}}\right\vert_{{\boldsymbol{t}}=\hat{\boldsymbol{\beta}}_n}^\top\left(\boldsymbol{\beta}_0-\hat{\boldsymbol{\beta}}_n\right)\right.\\
&~~~~~~~~~~~~~\left.-\frac{1}{2}\left(\boldsymbol{\beta}_0-\hat{\boldsymbol{\beta}}_n\right)^\top\left.\frac{\partial^2 l\left(\boldsymbol{\omega}_{\boldsymbol{t}}\left(\hat{\boldsymbol{\beta}}_n,\hat{\Lambda}_n,\boldsymbol{h}_n\right);\right)}{\partial {\boldsymbol{t}}\partial {\boldsymbol{t}}^\top}\right\vert_{{\boldsymbol{t}}=\tilde{\boldsymbol{\beta}}_n}
\left(\boldsymbol{\beta}_0-\hat{\boldsymbol{\beta}}_n\right)\right\}.
\end{align*}
The last equality in the preceding display is based on Taylor's expansion and $\tilde{\boldsymbol{\beta}}_n$ is between $\boldsymbol{\beta}_0$ and $\hat{\boldsymbol{\beta}}_n$. Since ${\boldsymbol{t}}=\hat{\boldsymbol{\beta}}_n$ reaches the maximum of $\mathbb{P}_nl\left(\boldsymbol{\omega}_{\boldsymbol{t}}\left(\hat{\boldsymbol{\beta}}_n,\hat{\Lambda}_n,\boldsymbol{h}_n\right);\right)$, it is true that $$\mathbb{P}_n\left\{\left.\frac{\partial l\left(\boldsymbol{\omega}_{\boldsymbol{t}}\left(\hat{\boldsymbol{\beta}}_n,\hat{\Lambda}_n,\boldsymbol{h}_n\right);\right)}{\partial {\boldsymbol{t}}}\right\vert_{{\boldsymbol{t}}=\hat{\boldsymbol{\beta}}_n}\right\}=0.$$ Therefore, we have
\begin{align*}
lr(\boldsymbol{\beta}_0)
\le -n\mathbb{P}_n\left\{\left(\boldsymbol{\beta}_0-\hat{\boldsymbol{\beta}}_n,\boldsymbol{h}_n\right)^\top\left.\frac{\partial^2 l\left(\boldsymbol{\omega}_{\boldsymbol{t}}\left(\hat{\boldsymbol{\beta}}_n,\hat{\Lambda}_n,\boldsymbol{h}_n\right);\right)}{\partial {\boldsymbol{t}}\partial {\boldsymbol{t}}^\top}\right\vert_{{\boldsymbol{t}}=\tilde{\boldsymbol{\beta}}_n}
\left(\boldsymbol{\beta}_0-\hat{\boldsymbol{\beta}}_n\right)\right\}.
\end{align*}
On the other hand, by (\ref{fullinequality}) and using the Taylor's expansion, we have
\begin{align*}
lr(\boldsymbol{\beta}_0)
&=2n\mathbb{P}_n\left\{l\left(\hat{\boldsymbol{\beta}}_n,\hat{\Lambda}_n;\right)-l\left(\boldsymbol{\beta}_0,\hat{\Lambda}_{n0};\right)\right\}\\
&\ge2n\mathbb{P}_n\left\{l\left(\boldsymbol{\omega}_{\hat{\boldsymbol{\beta}}_n}\left(\boldsymbol{\beta}_0, \hat{\Lambda}_{n0},\boldsymbol{h}_n\right);\right)-l\left(\boldsymbol{\omega}_{\boldsymbol{\beta}_0}\left(\boldsymbol{\beta}_0, \hat{\Lambda}_{n0},\boldsymbol{h}_n\right);\right)\right\}\\
&=2n\mathbb{P}_n\left\{\left.\frac{\partial l\left(\boldsymbol{\omega}_{\boldsymbol{t}}\left(\boldsymbol{\beta}_0,\hat{\Lambda}_{n0},\boldsymbol{h}_n\right);\right)}{\partial {\boldsymbol{t}}}\right\vert_{{\boldsymbol{t}}=\boldsymbol{\beta}_0}^\top\left(\hat{\boldsymbol{\beta}}_n-\boldsymbol{\beta}_0\right)\right\}\\
&~~~~+n\mathbb{P}_n\left\{\left(\hat{\boldsymbol{\beta}}_n-\boldsymbol{\beta}_0\right)^\top\left.\frac{\partial^2 l\left(\boldsymbol{\omega}_{\boldsymbol{t}}\left(\boldsymbol{\beta}_0,\hat{\Lambda}_{n0},\boldsymbol{h}_n\right);\right)}{\partial {\boldsymbol{t}}\partial {\boldsymbol{t}}^\top}\right\vert_{{\boldsymbol{t}}=\bar{\boldsymbol{\beta}}_n}
\left(\hat{\boldsymbol{\beta}}_n-\boldsymbol{\beta}_0\right)\right\},
\end{align*}
where $\bar{\boldsymbol{\beta}}_n$ is between $\boldsymbol{\beta}_0$ and $\hat{\boldsymbol{\beta}}_n$.
Denote
\begin{equation}\label{flnformula}
\begin{split}
f_{ln}&\equiv2n\mathbb{P}_n\left\{\left.\frac{\partial l\left(\boldsymbol{\omega}_{\boldsymbol{t}}\left(\boldsymbol{\beta}_0,\hat{\Lambda}_{n0},\boldsymbol{h}_n\right);\right)}{\partial {\boldsymbol{t}}}\right\vert_{{\boldsymbol{t}}=\boldsymbol{\beta}_0}^\top\left(\hat{\boldsymbol{\beta}}_n-\boldsymbol{\beta}_0\right)\right\}\\
&~~~~+n\mathbb{P}_n\left\{\left(\hat{\boldsymbol{\beta}}_n-\boldsymbol{\beta}_0\right)^\top\left.\frac{\partial^2 l\left(\boldsymbol{\omega}_{\boldsymbol{t}}\left(\boldsymbol{\beta}_0,\hat{\Lambda}_{n0},\boldsymbol{h}_n\right);\right)}{\partial {\boldsymbol{t}}\partial {\boldsymbol{t}}^\top}\right\vert_{{\boldsymbol{t}}=\bar{\boldsymbol{\beta}}_n}
\left(\hat{\boldsymbol{\beta}}_n-\boldsymbol{\beta}_0\right)\right\}
\end{split}
\end{equation}
and
\begin{equation}\label{frnformula}
f_{rn}\equiv-n\mathbb{P}_n\left\{\left(\boldsymbol{\beta}_0-\hat{\boldsymbol{\beta}}_n\right)^\top\left.\frac{\partial^2 l\left(\boldsymbol{\omega}_{\boldsymbol{t}}\left(\hat{\boldsymbol{\beta}}_n,\hat{\Lambda}_n,\boldsymbol{h}_n\right);\right)}{\partial {\boldsymbol{t}}\partial {\boldsymbol{t}}^\top}\right\vert_{{\boldsymbol{t}}=\tilde{\boldsymbol{\beta}}_n}
\left(\boldsymbol{\beta}_0-\hat{\boldsymbol{\beta}}_n\right)\right\}.
\end{equation}

By Lemma~1 and the asymptotic convergence of $\hat{\boldsymbol{\beta}}_n$, we have
\begin{align*}
&\Pr\left(\left\|\hat{\boldsymbol{\beta}}_n-\boldsymbol{\beta}_0\right\|<\epsilon_0, b_1<\hat{\lambda}_{n}<b_2, b_1<\hat{\lambda}_{n0}<b_,
\right.\\
&\left.~~~~~~~~~~~~~~~~~~~~~~~~~~~~~~\hat{\Lambda}_n(\tau_0)<(2M/3)\tau_0,\hat{\Lambda}_{n0}(\tau_0)<(2M/3)\tau_0
\right)
\to1.
\end{align*}
Besides, we have just proved that if
$b_1<\hat{\lambda}_{n}<b_2$, $b_1<\hat{\lambda}_{n0}<b_2$,
$\hat{\Lambda}_n(\tau_0)<(2M/3)\tau_0$, $\hat{\Lambda}_{n0}(\tau_0)<(2M/3)\tau_0$ and $\left\|\hat{\boldsymbol{\beta}}_n-\boldsymbol{\beta}_0\right\|<\epsilon_0$ for a sufficiently small $\epsilon_0>0$,
then it is true that $f_{ln}\le lr(\boldsymbol{\beta}_0) \le f_{rn}$ when $n$ is sufficiently large.

Hence, the preceding convergence implies that
\begin{equation}\label{inbetween}
\Pr\left(f_{ln}\le lr(\boldsymbol{\beta}_0) \le f_{rn}\right)\to 1.
\end{equation}

Finally, by Lemma~2 both $f_{ln}$ and $f_{rn}$ converge in distribution to the $\chi^2$-distribution with $K$ degrees of freedom. It can be shown that $lr(\boldsymbol{\beta}_0)$ also converges in distribution to the $\chi^2$-distribution with $K$ degrees of freedom. In what follows we prove this claim.

For any constant $t$
\begin{align*}
\Pr(f_{rn} \le t |lr(\boldsymbol{\beta}_0)\le f_{rn})\le \Pr(lr(\boldsymbol{\beta}_0)\le t|lr(\boldsymbol{\beta}_0)\le f_{rn}).
\end{align*}
It follows that
\begin{equation}\label{keyinequality}
\begin{split}
&\Pr(f_{rn} \le t |lr(\boldsymbol{\beta}_0)\le f_{rn}) \Pr(lr(\boldsymbol{\beta}_0)\le f_{rn})\\
&~~~~~~~~~~~~~~~~~~~~~~~~~\le \Pr(lr(\boldsymbol{\beta}_0)\le t|lr(\boldsymbol{\beta}_0)\le f_{rn}) \Pr(lr(\boldsymbol{\beta}_0)\le f_{rn}).
\end{split}
\end{equation}
By (\ref{inbetween}), for any small $\eta>0$, there exists $N$, such that
$$\Pr\left(lr(\boldsymbol{\beta}_0)\le f_{rn}\right)>1-\eta,$$
for $n>N$, which implies that
$$\Pr\left(lr(\boldsymbol{\beta}_0) > f_{rn}\right)<\eta.$$
Then
\begin{equation}\label{smallterm}
\Pr(f_{rn} \le t |lr(\boldsymbol{\beta}_0)> f_{rn}) \Pr(lr(\boldsymbol{\beta}_0)> f_{rn}) < \eta
\end{equation}
Adding (\ref{keyinequality}) by (\ref{smallterm}), we have
\begin{align*}
\Pr(f_{rn} \le t) &< \Pr(lr(\boldsymbol{\beta}_0)\le t|lr(\boldsymbol{\beta}_0)\le f_{rn}) \Pr(lr(\boldsymbol{\beta}_0)\le f_{rn}) + \eta\\
&< \Pr(lr(\boldsymbol{\beta}_0)\le t) +\eta.
\end{align*}
Since $\eta$ can be any small positive number, by contradiction we can show
$$\lim_{n\to\infty}\Pr(f_{rn}\le t)\le \liminf_{n\to\infty}\Pr(lr(\boldsymbol{\beta}_0)\le t).$$
Similarly, we can show that
\begin{align*}
\Pr(lr(\boldsymbol{\beta}_0)\le t)< \Pr(f_{ln}\le t)+\eta,
\end{align*}
when $n$ is large enough. Then
$$\limsup_{n\to\infty}\Pr(lr(\boldsymbol{\beta}_0)\le t)\le \lim_{n\to\infty}\Pr(f_{rn}\le t).$$

Since $f_{ln}$ and $f_{rn}$ both converge in distribution to the $\chi^2$-distribution with $K$ degrees of freedom and assume that $\phi$ follows this distribution. We have,
$$\lim_{n\to\infty}\Pr(f_{ln}\le t)=\lim_{n\to\infty}\Pr(f_{rn}\le t)=\Pr(\phi\le t).$$
Then, it is obvious that
$$\lim_{n\to\infty}\Pr(lr(\boldsymbol{\beta}_0)\le t)=\Pr(\phi\le t),$$
that is, $lr(\boldsymbol{\beta}_0)$ converges in distribution to the $\chi^2$-distribution with $K$ degrees of freedom. This completes the proof.~~~$\square$

\subsection*{Lemma~1 and its proof}
\begin{lem}\label{lemma1}
   Under conditions C1--C4, let $\hat{\lambda}_n$ and $\hat{\lambda}_{n0}$ be defined by~(\ref{derilambda}) in the proof of Theorem~1. 
   Assume that $b_1 = b_l/2$, $b_2 = (3b_u)/2$, and $M = 2b_u$, where $b_l$ and $b_u$, established in the proof of Theorem~1, are the minimum and maximum values, respectively, of the true hazard function $\lambda_0$ on $[0,\tau_1]$. Then
   \begin{align*}
   \Pr\!\left(b_1 < \hat{\lambda}_{n} < b_2\right) \to 1, \qquad 
   \Pr\!\left(b_1 < \hat{\lambda}_{n0} < b_2\right) \to 1,
   \end{align*}
   and
   \begin{align*}
   \Pr\!\left(\hat{\Lambda}_n(\tau_0) < \tfrac{2M}{3}\tau_0\right) \to 1, \qquad 
   \Pr\!\left(\hat{\Lambda}_{n0}(\tau_0) < \tfrac{2M}{3}\tau_0\right) \to 1.
   \end{align*}
\end{lem}

{\em Proof.}
In this proof and the proof of Lemma~\ref{lemma2}, $c$ denotes a generic constant that may vary between occurrences. 

Here, only the arguments for $\hat{\lambda}_n$ are presented ahead and those for $\hat{\lambda}_{n0}$ are very similar, thus omitted. First, we show that
\begin{align}\label{b1b2}
\Pr\left(b_1<\hat{\lambda}_n<b_2\right)\to1.
\end{align}

 By the Jackson-type theorem in~\citet{deboor:01} and condition C1, there exists $\lambda_n$, such that $\left\|\lambda_n-\lambda_0\right\|_\infty=O\left(n^{-\frac{p+r}{2(p+r)+1}}\right)$, where $\|\cdot\|_\infty$ for $\lambda_n-\lambda_0$ is in $[0,\tau_1]$. 
If $\Lambda_n=\int\lambda_n$, by $b_u<M$, it is clear that $\Lambda_n\in\mathcal{F}_n$ when restricted on $[\tau_0,\tau_1]$, for $\mathcal{F}_n$ described in the proof of Theorem~1.  In addition, $\left\|\Lambda_n-\Lambda_0\right\|_\infty\le \tau_1\left\|\lambda_n-\lambda_0\right\|_\infty=O\left(n^{-\frac{p+r}{2(p+r)+1}}\right)$, where $\|\cdot\|_\infty$ for $\Lambda_n-\Lambda_0$ is on $[\tau_0,\tau_1]$.

  Then we have
\begin{equation}\label{dishatlambda}
d\left(\hat{\Lambda}_n,\Lambda_n\right)\le d\left(\hat{\Lambda}_n,\Lambda_0\right)+d\left(\Lambda_n,\Lambda_0\right)=O_P\left(n^{-\frac{p+r}{2(p+r)+1}}\right),
\end{equation}
by $d\left(\hat{\Lambda}_n,\Lambda_0\right)=O_P\left(n^{-\frac{p+r}{2(p+r)+1}}\right)$.
Denote $f_U$ and $f_V$ as the density function of $U$ and $V$, respectively. Then by condition C2,
\begin{equation}\label{dtointegral}
\begin{split}
d^2\left(\hat{\Lambda}_n,\Lambda_n\right)
=&\int_{\tau_0}^{\tau_1}\left|\hat{\Lambda}_n(t)-\Lambda_n(t)\right|^2f_U(t)dt
+\int_{\tau_0}^{\tau_1}\left|\hat{\Lambda}_n(t)-\Lambda_n(t)\right|^2f_V(t)dt\\
&>\int_{\tau_0}^{\tau_1}\left|\hat{\Lambda}_n(t)-\Lambda_n(t)\right|^2dt.
\end{split}
\end{equation}

Let $t^\ast$ satisfy $\left|\hat{\Lambda}_n(t^\ast)-\Lambda_n(t^\ast)\right|=\left\|\hat{\Lambda}_n-\Lambda_n\right\|_\infty\equiv\xi$ with $\tau_0\le t^\ast \le \tau_1$. In addition, without loss of generality, assume $\hat{\Lambda}_n(t^\ast)-\Lambda_n(t^\ast)>0$ and $t^\ast <\tau_1$. By both $\hat{\Lambda}_n$ and $\Lambda_n$ belonging to $\mathcal{F}_n$, for $\mathcal{F}_n$ described in the proof of Theorem~1, we know that
$\left|\frac{\partial \left\{\hat{\Lambda}_n(t)-\Lambda_n(t)\right\}}{\partial t}\right|<2M$ on $[\tau_0,\tau_1]$. Then for any $s$ with $s>0$ and $t^\ast+s\le\tau_1$, we have
\begin{align*}
\left\{\hat{\Lambda}_n(t^\ast+s)-\Lambda_n(t^\ast+s)\right\}
-\left\{\hat{\Lambda}_n(t^\ast)-\Lambda_n(t^\ast)\right\}
\ge -2Ms.
\end{align*}
Hence, if $0\le s\le c_0\xi\le\frac{\xi}{4M}$ for some $c_0$ with $c_0>0$ and $t^\ast+c_0\xi\le \tau_1$, we have
\begin{align*}
\left\{\hat{\Lambda}_n(t^\ast+s)-\Lambda_n(t^\ast+s)\right\}\ge\xi-2Ms\ge\xi-2Mc_0\xi\ge\xi/2.
\end{align*}
Then
\begin{align*}
\int_{\tau_0}^{\tau_1}\left|\hat{\Lambda}_n(t)-\Lambda_n(t)\right|^2dt
\ge\int_{t^\ast}^{t^\ast+c_0\xi}\left|\hat{\Lambda}_n(t)-\Lambda_n(t)\right|^2dt\ge\frac{c_0}{4}\xi^3,
\end{align*}
which implies
\begin{align}\label{l2toinfty}
\left\|\hat{\Lambda}_n-\Lambda_n\right\|_\infty\le c d^{\frac{2}{3}}\left(\hat{\Lambda}_n,\Lambda_n\right),
\end{align}
by (\ref{dtointegral}). 
If $r_0>0$, by (\ref{dishatlambda}) we have $$d\left(\hat{\Lambda}_n,\Lambda_n\right)= o_P\left(n^{-\frac{p+r-r_0}{2(p+r)+1}}\right).$$ Then for any $\epsilon>0$, we have
$
\Pr\left(d\left(\hat{\Lambda}_n,\Lambda_n\right)<\epsilon n^{-\frac{p+r-r_0}{2(p+r)+1}}\right)\to1.
$
Hence, by (\ref{l2toinfty}) we have
\begin{align}\label{hatlambdato1}
 \Pr\left(\left\|\hat{\Lambda}_n-\Lambda_n\right\|_\infty<c n^{-\frac{2(p+r-r_0)/3}{2(p+r)+1}}\right)\to1
\end{align}
Assume $\hat{\Lambda}_n-\Lambda_n=\sum_{i=1}^{p_n}\alpha_iB_i^l$ with $p_n=O\left(n^{-\frac{1}{2(p+r)+1}}\right)$, then by~\citet{schumaker:2007} or~\citet{wu:thesis:2010}, its derivative
\begin{align}\label{hatlambdalambda}
\hat{\lambda}_n-\lambda_n=\sum_{i=1}^{p_n-1}\frac{(l-1)(\alpha_{i+1}-\alpha_i)}{u_{i+l}-u_i}B_{i+1}^{l-1},
\end{align}
where $\{B_i^l\}_{i=1}^{p_n}$ and $\{B_{i+1}^{l-1}\}_{i=1}^{p_n-1}$ are the B-spline basis functions of order $l$ and order $l-1$, respectively, $\{u_1,\cdots, u_{p_n+l}\}$ are the knot sequence for $\{B_i^l\}_{i=1}^{p_n}$ including the interior knots $\{u_i\}_{i=l+1}^{p_n}$.

If $\left\|\hat{\Lambda}_n-\Lambda_n\right\|_\infty<c n^{-\frac{2(p+r-r_0)/3}{2(p+r)+1}}$, from the fact that the B-spline basis functions are uniformly bounded and strongly localized~\citep{deboor:01}, it can be shown that
\begin{align*}
\alpha_i\le c n^{-\frac{2(p+r-r_0)/3}{2(p+r)+1}},
\end{align*}
for $i\in\{1,\cdots,q_n\}$.

On the other hand, we know that $\max_{i\in\{l,\cdots,p_n\}}(u_{i+1}-u_i)\ge cn^{-\frac{1}{2(p+r)+1}}$, since the number of subintervals in the B-spline knot sequence is $O\left(n^{\frac{1}{2(p+r)+1}}\right)$.
In addition, it is true that $\frac{\max_{i\in\{l,\cdots,p_n\}}(u_{i+1}-u_i)}{\min_{i\in\{l,\cdots,p_n\}}(u_{i+1}-u_i)}\le b_{knot}$ for a constant $b_{knot}>0$, from the fact that the B-spline knot sequence has interior knots placed uniformly or nearly uniformly. Then we have
$$u_{i+l}-u_i\ge\min_{i\in\{l,\cdots,p_n\}}(u_{i+1}-u_i)\ge cn^{-\frac{1}{2(p+r)+1}}.$$
Hence, if $2(p+r-r_0)/3-1>0$, that is, $r_0<(p+r)-3/2$, by the two preceding displays, for $i\in\{1,\cdots,p_n-1\}$ we have
\begin{align*}
\frac{(l-1)(\alpha_{i+1}-\alpha_i)}{u_{i+l}-u_i} \le cn^{-\frac{2(p+r-r_0)/3-1}{2(p+r)+1}}\to0.
\end{align*}
Then, (\ref{hatlambdalambda}) results in $\left\|\hat{\lambda}_n-\lambda_n\right\|_\infty\to0$. Hence, by $\left\|\lambda_n-\lambda_0\right\|_\infty=O\left(n^{-\frac{p+r}{2(p+r)+1}}\right)$, we have
$\left\|\hat{\lambda}_n-\lambda_0\right\|_\infty\to0$. Also by $b_l<\lambda_0<b_u$, for sufficiently large $n$, it is true that $$b_l/2<\hat{\lambda}_n<3b_u/2.$$

We have just established that,
$\left\|\hat{\Lambda}_n-\Lambda_n\right\|_\infty<c n^{-\frac{2(p+r-r_0)/3}{2(p+r)+1}}$ implies $b_l/2<\hat{\lambda}_n<3b_u/2$, for $r_0=\frac{(p+r)-3/2}{2}<(p+r)-3/2$, since $(p+r)-3/2>0$ by condition C1. In addition, (\ref{hatlambdato1}) holds by $r_0=\frac{(p+r)-3/2}{2}>0$. Hence, (\ref{b1b2}) holds.


In what follows, we complete the whole proof by showing
\begin{align}\label{lambdatau}
\Pr\left(\hat{\Lambda}_n(\tau_0)<(2M/3)\tau_0\right)\to1.
\end{align}
By (\ref{hatlambdato1}), that is,
$
 \Pr\left(\left\|\hat{\Lambda}_n-\Lambda_0\right\|_\infty<c n^{-\frac{2(p+r-r_0)/3}{2(p+r)+1}}\right)\to1,
$
we can show that for any $\epsilon>0$, it is true that
\begin{align*}
\Pr\left(\hat{\Lambda}_n(\tau_0)<\Lambda_0(\tau_0)+\epsilon\right)\to1.
\end{align*}
On the other hand, by $\Lambda_0(\tau_0)<b_u\tau_0$ and $M=2b_u$, we can choose a sufficiently small $\epsilon$ to satisfy
$$
\Pr\left(\hat{\Lambda}_n(\tau_0)<\Lambda_0(\tau_0)+\epsilon\right)<\Pr\left(\hat{\Lambda}_n(\tau_0)<(2M/3)\tau_0\right).
$$
Hence, (\ref{lambdatau}) holds.~~~$\square$
\subsection*{Lemma~2 and its proof}
\begin{lem}\label{lemma2}
   Under conditions C1--C4, $f_{ln}$ and $f_{rn}$ both converge in distribution to the $\chi^2$ distribution with $K$ degrees of freedom, 
   where $f_{ln}$ and $f_{rn}$ are defined by~(\ref{flnformula}) and~(\ref{frnformula}) in the proof of Theorem~1, respectively.
\end{lem}

{\em Proof.}
First, we show that
\begin{equation}\label{linearterm}
\sqrt{n}\mathbb{P}_n\left\{\left.\frac{\partial l\left(\boldsymbol{\omega}_{\boldsymbol{t}}\left(\boldsymbol{\beta}_0,\hat{\Lambda}_{n0},\boldsymbol{h}_n\right);\right)}{\partial {\boldsymbol{t}}}\right\vert_{{\boldsymbol{t}}=\boldsymbol{\beta}_0}
-\tilde{l}\left(\boldsymbol{\beta}_0,\Lambda_0;\right)\right\}\to_P{\boldsymbol{0}}.
\end{equation}

Let $\mathcal{F}=\{F:F\le M\tau_1\}$ for $M$ described in Lemma~1, for each $F\in\mathcal{F}$ restricted on $[\tau_0,\tau_1]$ is nonnegative and nondecreasing. $\boldsymbol{\mathcal{H}}=\left\{\boldsymbol{h}_n\right\}_{n>N_0}$ is the B-spline function sequence approximating $\boldsymbol{h}_\ast$ (\ref{hstar}), described in the proof of Theorem~1.

Denote
\begin{align*}
&\boldsymbol{\mathcal{L}}_1=\left\{\left.\frac{\partial l\left(\boldsymbol{\omega}_{\boldsymbol{t}}\left(\boldsymbol{\beta},\Lambda,\boldsymbol{g}\right);\right)}{\partial {\boldsymbol{t}}}\right\vert_{{\boldsymbol{t}}={\boldsymbol{t}}_0}:\right.\\
&\left.~~~~~~~~~~~~~~~~~~~~{\boldsymbol{t}}_0\in\boldsymbol{\mathcal{B}}, \Lambda_{{\boldsymbol{t}}_0}(\boldsymbol{\beta},\Lambda,\boldsymbol{g})\in\mathcal{F}, \boldsymbol{g}=(g_1,\cdots,g_K)^\top, \boldsymbol{g}\in\boldsymbol{\mathcal{H}}\right\},
\end{align*}
where $\boldsymbol{\mathcal{B}}$ is compact set in $\mathbb{R}^k$ and includes $\boldsymbol{\beta}_0$ in its interior. Since $\|h_{nk}-h_{\ast k}\|_\infty=O\left(n^{-\nu}\right)$, For any small $\epsilon>0$, there exists an $N$, such that $\|h_{nk}-h_{\ast k}\|_\infty<\epsilon/3$, for $n>N$ and $k\in\{1,\cdots,K\}$. Let $h_{low}=h_{\ast k}-\epsilon/3$ and $h_{high}=h_{\ast k}+\epsilon/3$. Then for any $n>N$ and $y$ with $\tau_0\le y \le \tau_1$, we have
$$h_{nk}(y)-\left(h_{\ast k}(y)-\epsilon/3\right)
\ge-\|h_{nk}-h_{\ast k}\|_\infty+\epsilon/3>0
$$
and
$$h_{nk}(y)-\left(h_{\ast k}(y)+\epsilon/3\right)
\le|h_{nk}-h_{\ast k}\|_\infty-\epsilon/3<0.
$$
In addition, $\|h_{high}-h_{low}\|_\infty=2\epsilon/3<\epsilon$. We know that the $\epsilon$-bracketing number for $\boldsymbol{\mathcal{H}}$ is less than $c(N+1)^K$.

Also by evaluating the $\epsilon$-covering number for $\boldsymbol{\mathcal{B}}$ and the $\epsilon$-bracketing number for $\mathcal{F}$. We can obtain the $\epsilon$-bracketing number of $\boldsymbol{\mathcal{L}}_1$ as done in the proof of Theorem 1 in~\citet{wu:xu:19}, which implies that
$\boldsymbol{\mathcal{L}}_1$ is a Donsker class.

Since $\boldsymbol{\beta}_0\in\boldsymbol{\mathcal{B}}$,  $\Lambda_{\boldsymbol{\beta}_0}\left(\boldsymbol{\beta}_0,\hat{\Lambda}_{n0},\boldsymbol{h}_n\right)=\hat{\Lambda}_{n0}\in\mathcal{F}_n\subset\mathcal{F}$ and $\boldsymbol{h}_n\in\boldsymbol{\mathcal{H}}$. Then, it is true that
$\left.
\frac{\partial l\left(\boldsymbol{\omega}_{\boldsymbol{t}}\left(\boldsymbol{\beta}_0,\hat{\Lambda}_{n0},\boldsymbol{h}_n\right);\right)}
{\partial {\boldsymbol{t}}}
\right\vert_{{\boldsymbol{t}}=\boldsymbol{\beta}_0}
\in\boldsymbol{\mathcal{L}}_1$. It can be easily derived that
\begin{align*}
\left.\frac{\partial l\left(\boldsymbol{\omega}_{\boldsymbol{t}}\left(\boldsymbol{\beta}_0,\hat{\Lambda}_{n0},\boldsymbol{h}_n\right);\right)}{\partial {\boldsymbol{t}}}\right\vert_{{\boldsymbol{t}}=\boldsymbol{\beta}_0}
=l_{\boldsymbol{\beta}}(\boldsymbol{\beta}_0,\hat{\Lambda}_{n0};)-l_\Lambda(\boldsymbol{\beta}_0,\hat{\Lambda}_{n0};)\boldsymbol{h}_n.
\end{align*}
Denote $P_0$ as the probability measure for the true data distribution. By the preceding display and (\ref{efficientscore}), We have
\begin{align*}
&P_0\left\|\left.\frac{\partial l\left(\boldsymbol{\omega}_{\boldsymbol{t}}\left(\boldsymbol{\beta}_0,\hat{\Lambda}_{n0},\boldsymbol{h}_n\right);\right)}{\partial {\boldsymbol{t}}}\right\vert_{{\boldsymbol{t}}=\boldsymbol{\beta}_0}
-\tilde{l}\left(\boldsymbol{\beta}_0,\Lambda_0;\right)\right\|^2\\
&~~~~~~~~~~~~=P_0\left\|l_{\boldsymbol{\beta}}(\boldsymbol{\beta}_0,\hat{\Lambda}_{n0};)-l_\Lambda(\boldsymbol{\beta}_0,\hat{\Lambda}_{n0};)\boldsymbol{h}_n
-l_{\boldsymbol{\beta}}(\boldsymbol{\beta}_0,\Lambda_0;)
+l_\Lambda(\boldsymbol{\beta}_0,\Lambda_0;)\boldsymbol{h}_\ast\right\|^2\\
&~~~~~~~~~~~~\le cd^2\left(\hat{\Lambda}_{n0}, \Lambda_0\right)+c\sum_{k=1}^K\|h_{nk}-h_{\ast k}\|_\infty^2\to_P0.
\end{align*}
Hence, by Corollary 2.3.12 in~\citet{van:wellner:96}, we have
\begin{equation}\label{corollary2.3.12}
\sqrt{n}\left(\mathbb{P}_n-P_0\right)\left\{\left.\frac{\partial l\left(\boldsymbol{\omega}_{\boldsymbol{t}}\left(\boldsymbol{\beta}_0,\hat{\Lambda}_{n0},\boldsymbol{h}_n\right);\right)}{\partial {\boldsymbol{t}}}\right\vert_{{\boldsymbol{t}}=\boldsymbol{\beta}_0}
-\tilde{l}\left(\boldsymbol{\beta}_0,\Lambda_0;\right)\right\}\to_P{\boldsymbol{0}}.
\end{equation}

Next, we know that
\begin{equation}\label{p0firstderivative}
\begin{split}
&P_0\left\{\left.\frac{\partial l\left(\boldsymbol{\omega}_{\boldsymbol{t}}\left(\boldsymbol{\beta}_0,\hat{\Lambda}_{n0},\boldsymbol{h}_n\right);\right)}{\partial {\boldsymbol{t}}}\right\vert_{{\boldsymbol{t}}=\boldsymbol{\beta}_0}\right\}\\
&~~~~~~~~~~~~=\left(P_0-P_{\boldsymbol{\beta}_0,\hat{\Lambda}_{n0}}\right)\tilde{l}(\boldsymbol{\beta}_0,\Lambda_0;)\\
&~~~~~~~~~~~~~~~~+\left(P_0-P_{\boldsymbol{\beta}_0,\hat{\Lambda}_{n0}}\right)
\left\{\left.\frac{\partial l\left(\boldsymbol{\omega}_{\boldsymbol{t}}\left(\boldsymbol{\beta}_0,\hat{\Lambda}_{n0},\boldsymbol{h}_n\right);\right)}{\partial {\boldsymbol{t}}}\right\vert_{{\boldsymbol{t}}=\boldsymbol{\beta}_0}-\tilde{l}(\boldsymbol{\beta}_0,\Lambda_0;)\right\},
\end{split}
\end{equation}
where $P_{\boldsymbol{\beta}_0,\hat{\Lambda}_{n0}}$ is a valid probability measure if $b_1<\hat{\lambda}_{n0}<b_2$ with $\hat{\Lambda}_{n0}=\int\hat{\lambda}_{n0}$.
The preceding display (\ref{p0firstderivative}) is due to
\begin{align*}
P_{\boldsymbol{\beta}_0,\hat{\Lambda}_{n0}}\left\{\left.\frac{\partial l\left(\boldsymbol{\omega}_{\boldsymbol{t}}\left(\boldsymbol{\beta}_0,\hat{\Lambda}_{n0},\boldsymbol{h}_n\right);\right)}{\partial {\boldsymbol{t}}}\right\vert_{{\boldsymbol{t}}=\boldsymbol{\beta}_0}\right\}=
&\int\left\{\left.\frac{\partial L\left(\boldsymbol{\omega}_{\boldsymbol{t}}\left(\boldsymbol{\beta}_0,\hat{\Lambda}_{n0},\boldsymbol{h}_n\right);\right) }{\partial {\boldsymbol{t}}}\right\vert_{{\boldsymbol{t}}=\boldsymbol{\beta}_0}\right\}\\
=&\left.\frac{\partial \int p_{{\boldsymbol{t}}, \hat{\Lambda}_{n0}+(\boldsymbol{\beta}_0-{\boldsymbol{t}})^\top\boldsymbol{h}_n} }{\partial {\boldsymbol{t}}}\right\vert_{{\boldsymbol{t}}=\boldsymbol{\beta}_0}={\boldsymbol{0}},
\end{align*}
where $p_{{\boldsymbol{t}}, \hat{\Lambda}_{n0}+(\boldsymbol{\beta}_0-{\boldsymbol{t}})^\top\boldsymbol{h}_n}$ is the density of the probability measure $P_{{\boldsymbol{t}}, \hat{\Lambda}_{n0}+(\boldsymbol{\beta}_0-{\boldsymbol{t}})^\top\boldsymbol{h}_n}$, when ${\boldsymbol{t}}$ is in a sufficiently small neighborhood of $\boldsymbol{\beta}_0$, $\boldsymbol{h}_n$ has uniformly bounded derivatives for sufficiently large $n$s,  and $b_1<\hat{\lambda}_{n0}<b_2$ for $b_1$ and $b_2$ described in Lemma~\ref{lemma1}. Then it is obvious that $\int p_{{\boldsymbol{t}}, \hat{\Lambda}_{n0}+(\boldsymbol{\beta}_0-{\boldsymbol{t}})^\top\boldsymbol{h}_n} \equiv 1$.

For the $k$th component of the first term on the right hand side in (\ref{p0firstderivative}), we have
\begin{align*}
&\left|\left(P_0-P_{\boldsymbol{\beta}_0,\hat{\Lambda}_{n0}}\right)\tilde{l}_k(\boldsymbol{\beta}_0,\Lambda_0;)\right|\\
&~~~~~~~~~~~~=\left|P_0\tilde{l}_k(\boldsymbol{\beta}_0,\Lambda_0;)\left\{\left(p_0-p_{\boldsymbol{\beta}_0,\hat{\Lambda}_{n0}}\right)\bigg/p_0
+l_\Lambda\left(\boldsymbol{\beta}_0,\Lambda_0;\right)\left(\hat{\Lambda}_{n0}-\Lambda_0\right)\right\}\right|\\
&~~~~~~~~~~~~=\left|P_0\tilde{l}_k(\boldsymbol{\beta}_0,\Lambda_0;)\left[\left\{p_0-p_{\boldsymbol{\beta}_0,\hat{\Lambda}_{n0}}+\frac{\partial p_0}{\partial \Lambda}\left(\hat{\Lambda}_{n0}-\Lambda_0\right)\right\}\bigg/p_0
\right]\right|\\
&~~~~~~~~~~~~=\left|P_0\tilde{l}_k(\boldsymbol{\beta}_0,\Lambda_0;)\left[\left\{-\frac{1}{2}\frac{\partial^2 p_{\beta_0, \Lambda_0+\tilde{t}\left(\hat{\Lambda}_{n0}-\Lambda_0\right)}}{\partial \Lambda^2}\left(\hat{\Lambda}_{n0}-\Lambda_0\right)^2\right\}\bigg/p_0
\right]\right|\\
&~~~~~~~~~~~~\le c d^2\left(\hat{\Lambda}_{n0},\Lambda_0\right)
\end{align*}
where the first equality in the preceding display is due to
$P_0\tilde{l}_k(\boldsymbol{\beta}_0,\Lambda_0;)l_{\Lambda}(\boldsymbol{\beta}_0,\Lambda_0;)\left(\Lambda_0-\hat{\Lambda}_{n0}\right)=0$ by Theorem 11.1 in \citet{vander:98}, $p_0$ and $p_{\boldsymbol{\beta}_0,\hat{\Lambda}_{n0}}$ are densities for $P_0$ and $P_{\boldsymbol{\beta}_0,\hat{\Lambda}_{n0}}$; the third equality is based on the Taylor's expansion for some $\tilde{t}$ between $0$ and $1$. Note that $\Lambda_0+\tilde{t}\left(\hat{\Lambda}_{n0}-\Lambda_0\right)$ is a valid cumulative hazard function, which implies that $p_{\beta_0, \Lambda_0+\tilde{t}\left(\hat{\Lambda}_{n0}-\Lambda_0\right)}$ is a valid density.

Besides, for the $k$th component of the second term on the right hand side in (\ref{p0firstderivative}), we have
\begin{align*}
&\left|\left(P_0-P_{\boldsymbol{\beta}_0,\hat{\Lambda}_{n0}}\right)\left\{\left.\frac{\partial l\left(\boldsymbol{\omega}_{\boldsymbol{t}}\left(\boldsymbol{\beta}_0,\hat{\Lambda}_{n0},\boldsymbol{h}_n\right);\right)}{\partial {\boldsymbol{t}}}\right\vert_{{\boldsymbol{t}}=\boldsymbol{\beta}_0}-\tilde{l}(\boldsymbol{\beta}_0,\Lambda_0;)\right\}_k\right|\\
&~~~~~~~~~~~~=\left|P_0\left\{\left.\frac{\partial l\left(\boldsymbol{\omega}_{\boldsymbol{t}}\left(\boldsymbol{\beta}_0,\hat{\Lambda}_{n0},\boldsymbol{h}_n\right);\right)}{\partial {\boldsymbol{t}}}\right\vert_{{\boldsymbol{t}}=\boldsymbol{\beta}_0}-\tilde{l}(\boldsymbol{\beta}_0,\Lambda_0;)\right\}_k
\left\{\left(p_0-p_{\boldsymbol{\beta}_0,\hat{\Lambda}_{n0}}\right)/p_0\right\}\right|\\
&~~~~~~~~~~~~\le\left[P_0\left\{\left.\frac{\partial l\left(\boldsymbol{\omega}_{\boldsymbol{t}}\left(\boldsymbol{\beta}_0,\hat{\Lambda}_{n0},\boldsymbol{h}_n\right);\right)}{\partial {\boldsymbol{t}}}\right\vert_{{\boldsymbol{t}}=\boldsymbol{\beta}_0}-\tilde{l}(\boldsymbol{\beta}_0,\Lambda_0;)\right\}_k^2
P_0\left\{\left(p_0-p_{\boldsymbol{\beta}_0,\hat{\Lambda}_{n0}}\right)/p_0\right\}^2\right]^{1/2}\\
&~~~~~~~~~~~~\le \left\{c d\left(\hat{\Lambda}_{n0},\Lambda_0\right)+ c\|h_{nk}-h_{\ast k}\|_\infty\right\}c d\left(\hat{\Lambda}_{n0},\Lambda_0\right)\\
\end{align*}

The preceding two displays imply that, if  $b_1<\hat{\lambda}_{n0}<b_2$,
\begin{align*}
\left|P_0\left\{\left.\frac{\partial l\left(\boldsymbol{\omega}_{\boldsymbol{t}}\left(\boldsymbol{\beta}_0,\hat{\Lambda}_{n0},\boldsymbol{h}_n\right);\right)}{\partial {\boldsymbol{t}}}\right\vert_{{\boldsymbol{t}}=\boldsymbol{\beta}_0}\right\}_k\right|\le c d^2\left(\hat{\Lambda}_{n0},\Lambda_0\right)+c\|h_{nk}-h_{\ast k}\|_\infty d\left(\hat{\Lambda}_{n0},\Lambda_0\right),
\end{align*}
for sufficiently large $n$.
In addition,
\begin{align*}
c d^2\left(\hat{\Lambda}_{n0},\Lambda_0\right)=
O_p\left(n^{-\frac{2(p+r)}{2(p+r)+1}}\right)=o_P\left(n^{-1/2}\right)
\end{align*}
\begin{align*}
 c\|h_{nk}-h_{\ast k}\|_\infty d\left(\hat{\Lambda}_{n0},\Lambda_0\right)
=&O\left(n^{-\frac{1}{2(p+r)+1}}\right)O_p\left(n^{-\frac{p+r}{2(p+r)+1}}\right)\\
=&O_P\left(n^{-\frac{(p+r)+1}{2(p+r)+1}}\right)=o_P\left(n^{-1/2}\right).
\end{align*}
Then for any small positive $\epsilon$, we have
\begin{equation}\label{epsilon12}
\begin{split}
&\Pr\left[\left|\left.\sqrt{n}P_0\left\{\frac{\partial l\left(\boldsymbol{\omega}_{\boldsymbol{t}}\left(\boldsymbol{\beta}_0,\hat{\Lambda}_{n0},\boldsymbol{h}_n\right);\right)}{\partial {\boldsymbol{t}}}\right\vert_{{\boldsymbol{t}}=\boldsymbol{\beta}_0}\right\}_k\right|>\epsilon,
b_1<\hat{\lambda}_{n0}<b_2
\right]\\
&~~~~~~~~~~~~~~~~\le\Pr\left[\sqrt{n}\left\{c d^2\left(\hat{\Lambda}_{n0},\Lambda_0\right)+c\|h_{nk}-h_{\ast k}\|_\infty d\left(\hat{\Lambda}_{n0},\Lambda_0\right)\right\}>\epsilon\right]
\to0.
\end{split}
\end{equation}
We also know that,
$\Pr\left(b_1<\hat{\lambda}_{n0}<b_2\right)\to1$ by Lemma~\ref{lemma1}. Hence, (\ref{epsilon12}) implies that,
\begin{align*}
 \Pr\left[\left|\sqrt{n}P_0\left\{\left.\frac{\partial l\left(\boldsymbol{\omega}_{\boldsymbol{t}}\left(\boldsymbol{\beta}_0,\hat{\Lambda}_{n0},\boldsymbol{h}_n\right);\right)}{\partial {\boldsymbol{t}}}\right\vert_{{\boldsymbol{t}}=\boldsymbol{\beta}_0}\right\}_k\right|>\epsilon
\right]\to0.
\end{align*}
Therefore,
\begin{align}\label{npconvergeto0}
\sqrt{n}P_0\left\{\left.\frac{\partial l\left(\boldsymbol{\omega}_{\boldsymbol{t}}\left(\boldsymbol{\beta}_0,\hat{\Lambda}_{n0},\boldsymbol{h}_n\right);\right)}{\partial {\boldsymbol{t}}}\right\vert_{{\boldsymbol{t}}=\boldsymbol{\beta}_0}\right\}\to_P{\boldsymbol{0}}.
\end{align}
Note that we have showed that the condition $b_1<\hat{\lambda}_{n0}<b_2$ is not required for the preceding convergence (\ref{npconvergeto0}).

In addition, since $\lambda_0$ has a positive lower bound and the derivative of $\boldsymbol{h}_\ast$ is bounded (by Theorem 4.1 in~\cite{huang:wellner:1997}), we can prove that the expectation of the score operator~(\ref{score_op}) for $h=h_{\ast k}$ (one component of $\boldsymbol{h}_\ast$ defined by (\ref{hstar})), evaluated at the true parameter, is $0$. Specifically,
\begin{align*}
P_0\left\{l_\Lambda\left(\boldsymbol{\beta}_0,\Lambda_0;\right)h_{\ast k}\right\}
&=P_0\left\{\left.\frac{\partial}{\partial s}l\left(\boldsymbol{\beta}_0,\Lambda_0+sh_{\ast k};\right)\right\vert_{s=0}\right\}
=P_0\left\{\frac{\left.\frac{\partial}{\partial s}p_{\boldsymbol{\beta}_0,\Lambda_0+sh_{\ast k}}\right\vert_{s=0}}
{p_{\boldsymbol{\beta}_0,\Lambda_0}}\right\}\\
&=P_0\left\{\frac{\lim_{w\to0}\frac{p_{\boldsymbol{\beta}_0,\Lambda_0+wh_{\ast k}}- p_{\boldsymbol{\beta}_0,\Lambda_0}}{w}}
{p_{\boldsymbol{\beta}_0,\Lambda_0}}\right\}
=\lim_{w\to0}\frac{\int p_{\boldsymbol{\beta}_0,\Lambda_0+wh_{\ast k}}-\int p_{\boldsymbol{\beta}_0,\Lambda_0}}{w}\\
&=0,
\end{align*}
where being equal to 0 is due to the fact that $p_{\boldsymbol{\beta}_0,\Lambda_0+wh_{\ast k}}$ represents a valid density and hence $\int p_{\boldsymbol{\beta}_0,\Lambda_0+wh_{\ast k}}\equiv1$, given that $w$ is in a sufficiently small neighborhood of $0$. Then it is true that for the efficient score~(\ref{efficientscore}),
$P_0\tilde{l}(\boldsymbol{\beta}_0,\Lambda_0;)={\boldsymbol{0}}$.
Hence, by (\ref{corollary2.3.12}) and (\ref{npconvergeto0}) it is easy to see that (\ref{linearterm}) holds.

Next, we establish that
\begin{equation}\label{quadraticterm}
-\mathbb{P}_n\left\{\left.\frac{\partial^2 l\left(\boldsymbol{\omega}_{\boldsymbol{t}}\left(\boldsymbol{\beta}_0,\hat{\Lambda}_{n0},\boldsymbol{h}_n\right);\right)}{\partial {\boldsymbol{t}}\partial {\boldsymbol{t}}^\top}\right\vert_{{\boldsymbol{t}}=\bar{\boldsymbol{\beta}}_n}\right\}\to_P P_0\tilde{l}(\boldsymbol{\beta}_0,\Lambda_0;)\tilde{l}(\boldsymbol{\beta}_0,\Lambda_0;)^\top
\end{equation}

Since
\begin{equation}\label{split}
\begin{split}
&\left.\frac{\partial^2 l\left(\boldsymbol{\omega}_{\boldsymbol{t}}\left(\boldsymbol{\beta}_0,\hat{\Lambda}_{n0},\boldsymbol{h}_n\right);\right)}{\partial {\boldsymbol{t}}\partial {\boldsymbol{t}}^\top}\right\vert_{{\boldsymbol{t}}=\bar{\boldsymbol{\beta}}_n}  \\
&~~~~~~~~~~~~=\left.\frac{\frac{\partial^2 L\left(\boldsymbol{\omega}_{\boldsymbol{t}}\left(\boldsymbol{\beta}_0,\hat{\Lambda}_{n0},\boldsymbol{h}_n\right);\right)}{\partial {\boldsymbol{t}}\partial {\boldsymbol{t}}^\top}}
{L\left(\boldsymbol{\omega}_{\boldsymbol{t}}\left(\boldsymbol{\beta}_0,\hat{\Lambda}_{n0},\boldsymbol{h}_n\right);\right)}\right\vert_{{\boldsymbol{t}}=\bar{\boldsymbol{\beta}}_n}
-\left.\left\{\frac{\partial l\left(\boldsymbol{\omega}_{\boldsymbol{t}}\left(\boldsymbol{\beta}_0,\hat{\Lambda}_{n0},\boldsymbol{h}_n\right);\right)}{\partial {\boldsymbol{t}}}\right\}^{\otimes2}\right\vert_{{\boldsymbol{t}}=\bar{\boldsymbol{\beta}}_n},
\end{split}
\end{equation}
where ${\boldsymbol{w}}^{\otimes2}$ is the outer product for column vector ${\boldsymbol{w}}$, that is, ${\boldsymbol{w}}^{\otimes2}={\boldsymbol{w}}{\boldsymbol{w}}^\top$. To prove (\ref{quadraticterm}), we deal with the two terms on the right hand side in the preceding display separately.
Denote
\begin{align*}
&\boldsymbol{\mathcal{L}}_2=\left\{\left.\frac{\frac{\partial^2 L\left(\boldsymbol{\omega}_{\boldsymbol{t}}\left(\boldsymbol{\beta},\Lambda,\boldsymbol{g}\right);\right)}{\partial {\boldsymbol{t}}\partial {\boldsymbol{t}}^\top}}
{L\left(\boldsymbol{\omega}_{\boldsymbol{t}}\left(\boldsymbol{\beta},\Lambda,\boldsymbol{g}\right);\right)}\right\vert_{{\boldsymbol{t}}={\boldsymbol{t}}_0}:\right.\\
&\left.~~~~~~~~~~~~~~~~~~~~{\boldsymbol{t}}_0\in\boldsymbol{\mathcal{B}}, \Lambda_{{\boldsymbol{t}}_0}(\boldsymbol{\beta},\Lambda,\boldsymbol{g})\in\mathcal{F}, \boldsymbol{g}=(g_1,\cdots,g_K)^\top, \boldsymbol{g}\in\boldsymbol{\mathcal{H}}\right\},
\end{align*}
where $\boldsymbol{\mathcal{B}}$, $\mathcal{F}$ and $\boldsymbol{\mathcal{H}}$ are as defined for $\boldsymbol{\mathcal{L}}_1$. We can show that $\boldsymbol{\mathcal{L}}_2$ is a Glivenko-Cantelli class by evaluating its bracketing numbers as did for $\boldsymbol{\mathcal{L}}_1$. That is,
\begin{equation}\label{GK}
\sup_{{\boldsymbol{f}}\in\boldsymbol{\mathcal{L}}_2}\left\|\left(\mathbb{P}_n-P_0\right){\boldsymbol{f}}\right\|\to_P0.
\end{equation}
If $\left\|\hat{\boldsymbol{\beta}}_n-\boldsymbol{\beta}_0\right\|<\epsilon_0$,
which implies $\left\|\bar{\boldsymbol{\beta}}_n-\boldsymbol{\beta}_0\right\|<\epsilon_0$.
Then  $\bar{\boldsymbol{\beta}}_n\in\boldsymbol{\mathcal{B}}$ since $\bar{\boldsymbol{\beta}}_n$ is between $\boldsymbol{\beta}_0$ and $\hat{\boldsymbol{\beta}}_n$.
And if $b_1<\hat{\lambda}_{n0}<b_2$ and $\hat{\Lambda}_{n0}(\tau_0)<(2M/3)\tau_0$ for $b_1$, $b_2$ and $M$ described in Lemma~\ref{lemma1}, we have
$\Lambda_{\bar{\boldsymbol{\beta}}_n}(\boldsymbol{\beta}_0,\hat{\Lambda}_{n0},\boldsymbol{h}_n)\in\mathcal{F}_n\subset\mathcal{F}$ for sufficiently large $n$, as argued in the proof of Theorem~1. Also by $\boldsymbol{h}_n\in\boldsymbol{\mathcal{H}}$, we have $\left.\frac{\frac{\partial^2 L\left(\boldsymbol{\omega}_{\boldsymbol{t}}\left(\boldsymbol{\beta}_0,\hat{\Lambda}_{n0},\boldsymbol{h}_n\right);\right)}{\partial {\boldsymbol{t}}\partial {\boldsymbol{t}}^\top}}
{L\left(\boldsymbol{\omega}_{\boldsymbol{t}}\left(\boldsymbol{\beta}_0,\hat{\Lambda}_{n0},\boldsymbol{h}_n\right);\right)}\right\vert_{{\boldsymbol{t}}=\bar{\boldsymbol{\beta}}_n}\in\boldsymbol{\mathcal{L}}_2$. Then for any small $\epsilon$, we have
\begin{align*}
&\Pr\left[\left\|\left(\mathbb{P}_n-P_0\right)\left\{\left.\frac{\frac{\partial^2 L\left(\boldsymbol{\omega}_{\boldsymbol{t}}\left(\boldsymbol{\beta}_0,\hat{\Lambda}_{n0},\boldsymbol{h}_n\right);\right)}{\partial {\boldsymbol{t}}\partial {\boldsymbol{t}}^\top}}
{L\left(\boldsymbol{\omega}_{\boldsymbol{t}}\left(\boldsymbol{\beta}_0,\hat{\Lambda}_{n0},\boldsymbol{h}_n\right);\right)}\right\vert_{{\boldsymbol{t}}=\bar{\boldsymbol{\beta}}_n}\right\}\right\|>\epsilon,
\left\|\hat{\boldsymbol{\beta}}_n-\boldsymbol{\beta}_0\right\|<\epsilon_0,
b_1<\hat{\lambda}_{n0}<b_2,\right.\\
&~~~~~~~~~~~~~~~~~~~~~~~~~~~~~~~~~~~~~~~~~~~~~~~~~~~~~~~~~~~~~~~~~~~~~~~~~~~~~~~~~~~~~~~~~~~
\left.\hat{\Lambda}_{n0}(\tau_0)<(2M/3)\tau_0
\right]\\
&~~~~\le\Pr\left[\left\|\left(\mathbb{P}_n-P_0\right)\left\{\left.\frac{\frac{\partial^2 L\left(\boldsymbol{\omega}_{\boldsymbol{t}}\left(\boldsymbol{\beta}_0,\hat{\Lambda}_{n0},\boldsymbol{h}_n\right);\right)}{\partial {\boldsymbol{t}}\partial {\boldsymbol{t}}^\top}}
{L\left(\boldsymbol{\omega}_{\boldsymbol{t}}\left(\boldsymbol{\beta}_0,\hat{\Lambda}_{n0},\boldsymbol{h}_n\right);\right)}\right\vert_{{\boldsymbol{t}}=\bar{\boldsymbol{\beta}}_n}\right\}\right\|>\epsilon,\left.\frac{\frac{\partial^2 L\left(\boldsymbol{\omega}_{\boldsymbol{t}}\left(\boldsymbol{\beta}_0,\hat{\Lambda}_{n0},\boldsymbol{h}_n\right);\right)}{\partial {\boldsymbol{t}}\partial {\boldsymbol{t}}^\top}}
{L\left(\boldsymbol{\omega}_{\boldsymbol{t}}\left(\boldsymbol{\beta}_0,\hat{\Lambda}_{n0},\boldsymbol{h}_n\right);\right)}\right\vert_{{\boldsymbol{t}}=\bar{\boldsymbol{\beta}}_n}\in\boldsymbol{\mathcal{L}}_2
\right]\\
&~~~~\le\Pr\left\{\sup_{{\boldsymbol{f}}\in\boldsymbol{\mathcal{L}}_2}\left\|\left(\mathbb{P}_n-P_0\right){\boldsymbol{f}}\right\|>\epsilon\right\}\to0,
\end{align*}
where the convergence to 0 in the preceding display is due to (\ref{GK}).

In addition, Lemma~\ref{lemma1} implies that $\Pr\left\{\left\|\hat{\boldsymbol{\beta}}_n-\boldsymbol{\beta}_0\right\|<\epsilon_0,
b_1<\hat{\lambda}_{n0}<b_2,\hat{\Lambda}_{n0}(\tau_0)<(2M/3)\tau_0\right\}\to1$. Hence, it can be shown that
\begin{equation*}
\left(\mathbb{P}_n-P_0\right)\left\{\left.\frac{\frac{\partial^2 L\left(\boldsymbol{\omega}_{\boldsymbol{t}}\left(\boldsymbol{\beta}_0,\hat{\Lambda}_{n0},\boldsymbol{h}_n\right);\right)}{\partial {\boldsymbol{t}}\partial {\boldsymbol{t}}^\top}}
{L\left(\boldsymbol{\omega}_{\boldsymbol{t}}\left(\boldsymbol{\beta}_0,\hat{\Lambda}_{n0},\boldsymbol{h}_n\right);\right)}\right\vert_{{\boldsymbol{t}}=\bar{\boldsymbol{\beta}}_n}\right\}\to_P{\boldsymbol{0}}.
\end{equation*}

In addition, since $d\left\{\left(\hat{\boldsymbol{\beta}}_n,\hat{\Lambda}_{n0}\right),\left(\boldsymbol{\beta}_0,\Lambda_0\right)\right\}\to_P0$, it is true that
\begin{align*}
P_0\left\{\left.\frac{\frac{\partial^2 L\left(\boldsymbol{\omega}_{\boldsymbol{t}}\left(\boldsymbol{\beta}_0,\hat{\Lambda}_{n0},\boldsymbol{h}_n\right);\right)}{\partial {\boldsymbol{t}}\partial {\boldsymbol{t}}^\top}}
{L\left(\boldsymbol{\omega}_{\boldsymbol{t}}\left(\boldsymbol{\beta}_0,\hat{\Lambda}_{n0},\boldsymbol{h}_n\right);\right)}\right\vert_{{\boldsymbol{t}}=\bar{\boldsymbol{\beta}}_n}
-\left.\frac{\frac{\partial^2 L\left(\boldsymbol{\omega}_{\boldsymbol{t}}\left(\boldsymbol{\beta}_0,\Lambda_0,\boldsymbol{h}_n\right);\right)}{\partial {\boldsymbol{t}}\partial {\boldsymbol{t}}^\top}}
{L\left(\boldsymbol{\omega}_{\boldsymbol{t}}\left(\boldsymbol{\beta}_0,\Lambda_0,\boldsymbol{h}_n\right);\right)}\right\vert_{{\boldsymbol{t}}=\boldsymbol{\beta}_0}
\right\}\to_P{\boldsymbol{0}},
\end{align*}
where
\begin{align*}
P_0\left\{\left.\frac{\frac{\partial^2 L\left(\boldsymbol{\omega}_{\boldsymbol{t}}\left(\boldsymbol{\beta}_0,\Lambda_0,\boldsymbol{h}_n\right);\right)}{\partial {\boldsymbol{t}}\partial {\boldsymbol{t}}^\top}}
{L\left(\boldsymbol{\omega}_{\boldsymbol{t}}\left(\boldsymbol{\beta}_0,\Lambda_0,\boldsymbol{h}_n\right);\right)}\right\vert_{{\boldsymbol{t}}=\boldsymbol{\beta}_0}\right\}
&=\int \left\{\left.\frac{\partial^2 L\left(\boldsymbol{\omega}_{\boldsymbol{t}}\left(\boldsymbol{\beta}_0,\Lambda_0,\boldsymbol{h}_n\right);\right)}{\partial {\boldsymbol{t}}\partial {\boldsymbol{t}}^\top}\right\vert_{{\boldsymbol{t}}=\boldsymbol{\beta}_0}\right\}\\
&=\left.\frac{\partial^2 \int p_{{\boldsymbol{t}},\Lambda_0+(\boldsymbol{\beta}_0-{\boldsymbol{t}})^\top\boldsymbol{h}_n }}{\partial {\boldsymbol{t}}\partial {\boldsymbol{t}}^\top} \right\vert_{{\boldsymbol{t}}=\boldsymbol{\beta}_0}={\boldsymbol{0}},
\end{align*}
by $\int p_{{\boldsymbol{t}},\Lambda_0+(\boldsymbol{\beta}_0-{\boldsymbol{t}})^\top\boldsymbol{h}_n }\equiv 1$ when ${\boldsymbol{t}}$ is in a sufficiently small neighborhood of  $\boldsymbol{\beta}_0$.
Hence, we have
\begin{align*}
\mathbb{P}_n\left\{\left.\frac{\frac{\partial^2 L\left(\boldsymbol{\omega}_{\boldsymbol{t}}\left(\boldsymbol{\beta}_0,\hat{\Lambda}_{n0},\boldsymbol{h}_n\right);\right)}{\partial {\boldsymbol{t}}\partial {\boldsymbol{t}}^\top}}
{L\left(\boldsymbol{\omega}_{\boldsymbol{t}}\left(\boldsymbol{\beta}_0,\hat{\Lambda}_{n0},\boldsymbol{h}_n\right);\right)}\right\vert_{{\boldsymbol{t}}=\bar{\boldsymbol{\beta}}_n}\right\}\to_P{\boldsymbol{0}}.
\end{align*}
Then by (\ref{split}), it is true that
\begin{equation}\label{secondlast}
-\mathbb{P}_n\left[\left.\frac{\partial^2 l\left(\boldsymbol{\omega}_{\boldsymbol{t}}\left(\boldsymbol{\beta}_0,\hat{\Lambda}_{n0},\boldsymbol{h}_n\right);\right)}{\partial {\boldsymbol{t}}\partial {\boldsymbol{t}}^\top}\right\vert_{{\boldsymbol{t}}=\bar{\boldsymbol{\beta}}_n}
+\left.\left\{\frac{\partial l\left(\boldsymbol{\omega}_{\boldsymbol{t}}\left(\boldsymbol{\beta}_0,\hat{\Lambda}_{n0},\boldsymbol{h}_n\right);\right)}{\partial {\boldsymbol{t}}}\right\}^{\otimes2}\right\vert_{{\boldsymbol{t}}=\bar{\boldsymbol{\beta}}_n}\right]\to_P{\boldsymbol{0}}.
\end{equation}

Next, by Example 19.20 of~\citet{vander:98} a Donsker class multiplied by  another Donsker class is a Glivenko-Cantelli class if the two Donsker classes are both uniformly bounded. Hence, we have
\begin{align*}
\left(\mathbb{P}_n-P_0\right)\left[\left.\left\{\frac{\partial l\left(\boldsymbol{\omega}_{\boldsymbol{t}}\left(\boldsymbol{\beta}_0,\hat{\Lambda}_{n0},\boldsymbol{h}_n\right);\right)}{\partial {\boldsymbol{t}}}\right\}^{\otimes2}\right\vert_{{\boldsymbol{t}}=\bar{\boldsymbol{\beta}}_n}\right]\to_P {\boldsymbol{0}}.
\end{align*}
Besides, since $d\left\{\left(\hat{\boldsymbol{\beta}}_n,\hat{\Lambda}_{n0}\right),\left(\boldsymbol{\beta}_0,\Lambda_0\right)\right\}\to_P0$ and $\left\|\boldsymbol{h}_n-\boldsymbol{h}_\ast\right\|_\infty\to0$, it is true that
\begin{align*}
P_0\left[\left.\left\{\frac{\partial l\left(\boldsymbol{\omega}_{\boldsymbol{t}}\left(\boldsymbol{\beta}_0,\hat{\Lambda}_{n0},\boldsymbol{h}_n\right);\right)}{\partial {\boldsymbol{t}}}\right\}^{\otimes2}\right\vert_{{\boldsymbol{t}}=\bar{\boldsymbol{\beta}}_n}
-\tilde{l}(\boldsymbol{\beta}_0,\Lambda_0;)\tilde{l}(\boldsymbol{\beta}_0,\Lambda_0;)^\top
\right]\to_P {\boldsymbol{0}}.
\end{align*}
The previous two displays result in
\begin{equation}\label{last}
\mathbb{P}_n\left.\left\{\frac{\partial l\left(\boldsymbol{\omega}_{\boldsymbol{t}}\left(\boldsymbol{\beta}_0,\hat{\Lambda}_{n0},\boldsymbol{h}_n\right);\right)}{\partial {\boldsymbol{t}}}\right\}^{\otimes2}\right\vert_{{\boldsymbol{t}}=\bar{\boldsymbol{\beta}}_n}
\to_P P_0\tilde{l}(\boldsymbol{\beta}_0,\Lambda_0;)\tilde{l}(\boldsymbol{\beta}_0,\Lambda_0;)^\top
\end{equation}
Adding (\ref{secondlast}) by (\ref{last}), we obtain (\ref{quadraticterm}).

The asymptotic normality (\ref{normality}), together with (\ref{linearterm}) and (\ref{quadraticterm}) implies that $f_{ln}$ converge in distribution to the $\chi^2$-distribution with $K$ degrees of freedom.

Using the same approach for showing (\ref{quadraticterm}), we can prove that
\begin{equation}\label{quadraticterm2}
-\mathbb{P}_n\left\{\left.\frac{\partial^2 l\left(\boldsymbol{\omega}_{\boldsymbol{t}}\left(\hat{\boldsymbol{\beta}}_n,\hat{\Lambda}_{n},\boldsymbol{h}_n\right);\right)}{\partial {\boldsymbol{t}}\partial {\boldsymbol{t}}^\top}\right\vert_{{\boldsymbol{t}}=\tilde{\boldsymbol{\beta}}_n}\right\}\to_P P_0\tilde{l}(\boldsymbol{\beta}_0,\Lambda_0;)\tilde{l}(\boldsymbol{\beta}_0,\Lambda_0;)^\top.
\end{equation}
Then, the asymptotic normality~(\ref{normality}) and the preceding display (\ref{quadraticterm2}) result in that $f_{rn}$ also converge in distribution to the $\chi^2$-distribution with $K$ degrees of freedom.  This completes the proof.~~~$\square$



\end{document}